\documentclass[12pt,noinfoline]{imsart}

\usepackage{amsmath, amssymb,amsthm, amscd, amsfonts}
\usepackage{graphicx, caption, subcaption}
\usepackage{fancyhdr}
\usepackage{geometry}
\usepackage{url,verbatim}
\usepackage{hyperref}
\usepackage{color}
\usepackage[ruled,  lined,  commentsnumbered]{algorithm2e}

\theoremstyle{plain}
\newtheorem{theorem}{Theorem}[section]
\newtheorem{corollary}[theorem]{Corollary}

\newtheorem{cor}[theorem]{Corollary}

\newtheorem{proposition}[theorem]{Proposition}
\newtheorem{prop}[theorem]{Proposition}
\theoremstyle{definition}
\newtheorem{defn}[theorem]{Definition}

\newtheorem{remark}[theorem]{Remark}
\newtheorem{rmk}[theorem]{Remark}
\newtheorem{example}[theorem]{Example}

\newtheorem{notation}{Notation}

\definecolor{forestgreen}{rgb}{0,.72,0} 
\definecolor{brickred}{rgb}{.72,0,0}

\newcommand{\Hp}{H_{p_1}}

\newcommand{\F}{\mathcal F}
\newcommand{\cH}{H}
\newcommand{\M}{\mathcal M}

\newcommand{\Z}{\mathbb Z}
\newcommand{\cB}{\mathcal B}
\newcommand{\cR}{\mathcal R}

\DeclareMathOperator\ci{\perp\!\!\!\perp} 

\newcommand{\suffstats}{{S=s}}
\newcommand{\suffstatsofu}{{S(U)=s(u)}}
\newcommand{\simplefiber}{\overline\F_{\suffstats}}
\newcommand{\fiber}{\F_{\suffstats}} 
\newcommand{\fiberofg}{\F_{S(g)=s(g)}} 
\newcommand{\fiberofu}{\F_{s(u)}} 
\newcommand{\modelM}{\mathcal M}

  \renewenvironment{thebibliography}[1]{
    \begin{oldthebibliography}{#1}
      \setlength{\parskip}{0.2ex}
      \setlength{\itemsep}{0.0in}
  }
  {
    \end{oldthebibliography}
  }

\begin{document}

\title{Goodness-of-fit for log-linear  network models: Dynamic Markov bases using hypergraphs}

\author{ Elizabeth Gross\thanks{\textit{eagross@ncsu.edu}, North Carolina State University 
 },
		 Sonja Petrovi\'c\thanks{\textit{Sonja.Petrovic@iit.edu}, Illinois Institute of Technology},
		 Despina Stasi\thanks{\textit{despina.stasi@gmail.com}, Pennsylvania State University and Illinois Institute of Technology} } 
\date{\today}

\begin{abstract}Social networks and other large sparse data sets pose significant challenges for statistical inference, as many standard statistical methods for testing model/data fit are not applicable in such settings. Algebraic statistics offers a theoretically justified approach to goodness-of-fit testing that relies on the theory of Markov bases and is intimately connected with the geometry of the model as described by its fibers. 

Most current practices require the computation of the entire basis, which is infeasible in many practical settings. We present a dynamic approach to explore the fiber of a model, which bypasses this issue, and is  based on the combinatorics of hypergraphs arising from the toric algebra structure of log-linear models.  

We demonstrate the approach on the Holland-Leinhardt $p_1$ model for  random directed graphs that allows for reciprocated edges.\end{abstract}

\maketitle
\section{Introduction}

Network data often arise as a single sparse observation of relationships among units, for example, individuals in a network of friendships, or species in a food web. 
Such a network can be naturally represented as a contingency table whose entries indicate the presence and type  of a relationship, and whose dimension depends on the complexity of the model. 
 This representation makes networks amenable to analysis by standard categorical data analysis tools and, in particular, it brings to bear the log-linear models literature, e.g.\ \cite{BFH75}. 
However, given that often only a small sample or even just a single observation of the network  is  all we have access to, or that the data are sparse, several problems remain. In particular, in the case of network models, since 
quantitative methods are lacking, goodness-of-fit testing  is usually carried out qualitatively using model diagnostics. Namely,  the clustering coefficient, triangle count, or another network characteristic is used for a heuristic comparison between observed and simulated data.  
In \cite{Hunter}, the authors offer a systematic approach for 
comparing structural statistics between an observed network and networks simulated from the fitted model, and point out some of the difficulties of fitting the ERGMs. More recently, \cite{F-review} review various network models and discuss modeling and fitting challenges that remain. 

Even for linear exponential families, the problem of determining goodness of fit is a difficult one for network data. 
When standard asymptotic methods, such as $\chi^2$ approximations, are deemed unreliable  (see \cite{Haberman}), or when the observed data are sparse, one may want to use exact conditional tests. In such tests, the observed network (or table) $u$ with sufficient statistics vector $\suffstatsofu$ is compared to 
 the reference set, called the \emph{fiber} $\fiber$, defined to be the space of all realizations of the network under the given set of constraints $\suffstats$. 
 Unfortunately, the size and combinatorial complexity  of the fiber are the main obstacle for complete fiber enumeration, so that even in small problems (e.g., see \cite[\S 4]{SlaZhuPet}), determining the exact distribution is often unfeasible. 
Moreover, fiber enumeration and sampling  is crucial not only for goodness-of-fit testing but also for data privacy considerations (see \cite{Sesa}). 

The theory of Markov bases provides a possible solution to the problem of sampling the fibers for any log-linear model. Namely, a Markov basis is a set of ``moves" that, starting from any point in a fiber,  allows one to perform a random walk on the fiber and visit every point with positive probability.  Therefore, the standard Metropolis-Hastings algorithm provides a way to carry out exact tests, and as argued in \cite{DS98}, this procedure yields {\it bona fide} tests for goodness of fit. 
Furthermore,  every log-linear model comes equipped with a non-unique but  finite Markov basis. The existence and finiteness of the basis is a  consequence of what is now often called the Fundamental Theorem of Markov Bases \cite{DS98} in the algebraic statistics literature.
However, two main computational challenges remain open to make this theory useful for network and large table data in practice.  We describe these challenges broadly next and, then, address them in the remainder of this manuscript.

\smallskip
The first computational challenge is in determining the Markov basis itself. The fact that  a Markov basis for a model guarantees to connect \emph{every} one of its fibers makes it a highly desirable object to obtain.  Unfortunately, the fastest algorithms for computing the moves for an arbitrary model (these algorithms exploit the toric structure of the  model) are not fast enough. Even for some basic log-linear network models, it can take hours to find all Markov moves for networks with less than 10 nodes. 
This motivates a structural study of 
 Markov bases for a given fixed family of models. To this end, the literature provides many examples \cite{AT03, AT05}, \cite{DS03}, \cite{Dob03}, \cite{DS04}, \cite{HAT10}, \cite{HTY09a, HTY09b}, \cite{HMTY13}, \cite{KNP10}, \cite{Nor12}, \cite{RY10}, \cite{SV12}, \cite{YOT13}. 
In addition, since our example of interest is a network model with inherent sampling constraints, we should note that such constraints can compound the issue of computing a set of moves guaranteed to connect each fiber.  Sampling constraints restrict the fiber, and in fact,  if one is interested in sampling a restricted fiber, \cite{OHT} and \cite{AHT2012} show that one needs a larger set of moves, for example a \emph{Graver basis}, 
to guarantee connectivity.  A Graver basis (see \cite[\S 1.3]{DSS09}, \cite[\S 4.6]{AHT2012} for definition and discussion) is a particular Markov basis and generally contains more moves than a minimal Markov basis (where minimal is defined with respect to set inclusion).

The second computational challenge comes from the fact that knowing an entire Markov basis for a model may still not be sufficient to run goodness-of-fit tests efficiently. Namely, Markov bases are data-independent; see Problem 5.5. in \cite{DobraEtAl-IMA}. 
 To paraphrase \cite{AHT2012}: since a Markov basis is common for every fiber $\fiber$ (that is, for all values $s$ that the vector $S$ of sufficient statistics can take), the set of moves connecting the particular fiber of the observed data   $u\in \fiberofu$ will usually be significantly smaller than the entire basis for the model.  To handle this issue Dobra, in \cite{Dobra2012}, suggests generating only moves needed to complete one step of the random walk, that is, only \emph{applicable} moves.  Dobra refers to the set of moves generated in this way as a \emph{dynamic Markov basis}, since the full basis is not generated ahead of time. 
An example of this strategy is  \cite{OHT}, where the authors present an 
algorithm for generating a random element of the Graver basis for the beta model.  The beta model is a basic generalization of the Erd\"os-Renyi random graph model: an ERGM for simple undirected random graphs where the degrees of the nodes form the sufficient statistics. 
In fact, this work can be cast within  a more general framework  of sampling from the space of  contingency tables with fixed properties. A commonly fixed set of table properties are marginals of the table: they represent sufficient statistics of many - but not all - log-linear models. 
 The paper \cite{Dobra2012} focuses on log-linear models whose sufficient statistics are fixed marginals. 
  There, the  Markov moves  are obtained through a sequential adjustment of cell bounds, a method that appears in sequential importance sampling (SIS) \cite{CDS05}, \cite{DC11}.  
In contrast, we build  a dynamic Markov basis by exploiting the combinatorics of the model. 
 This allows us to extend Dobra's methodology to log-linear models whose sufficient statistics are not necessarily table marginals.

\smallskip
In this manuscript, we explore the problem of  performing goodness-of-fit tests for log-linear models when sufficient statistics are not necessarily table marginals, and  
in the presence of sampling constraints.  
In this case, 
there is no general methodology for obtaining the part of the Markov bases which is relevant for the observed data. 
 In this work, we address the issues raised above from the point of view of algebraic statistics and combinatorial commutative algebra. We propose the use of \emph{parameter hypergraphs} to generate \emph{Graver} moves that are data-dependent and therefore applicable to the observed network (or table). Using Graver bases ensures connectivity of restricted fibers, while respecting sampling constraints. Furthermore, as \cite{PS13}  frame the Graver basis determination  problem in terms of combinatorics of hypergraphs, we add this combinatorial ingredient to the recipe which allows us to generate the moves in a dynamic fashion, based on the observed table or network. 
The sufficient statistics for the model need not be table marginals; the only assumption we impose, mostly for simplicity, is that the model parametrization is squarefree in the parameters (see Section~\ref{sec:hypergraph} for details). 
The random walk associated to the moves we produce in this way is irreducible, symmetric, and aperiodic, and so we may use the Metropolis-Hastings algorithm (see \cite[\S 7]{RC99}) to implement a Markov chain whose stationary distribution is equal to the conditional distribution on the fiber. This allows us to sample from the the fiber of an observed network or table as desired.

We illustrate our methodology and apply dynamically generated Markov bases to Holland and Leinhardt's  $p_1$ model  \cite{HL81}; specifically because previous methods are not applicable to this model directly. 
Holland and Leinhardt proposed to model a random directed graph by parametrizing propensity of nodes to send and receive links as well as reciprocate edges, where dyads are independent of each other. \cite{PRF10, FPRholland} study the algebra and geometry of these models and derive structural results for their Markov bases. Remarkably, the moves can be obtained by a direct computation only for networks with less than $7$ nodes, using 4ti2 \cite{4ti2}, currently the fastest software capable of producing such bases. 
Thus testing model fit for larger networks is not feasible using the traditional Metropolis-Hastings algorithm.  Using  a straightforward  implementation of Algorithm~\ref{alg:wrapper} in  R  \cite{R05}, we test several familiar  network data sets. 
 Figure \ref{fig:Monk-1MgofsHistogramExcl50kBurninSteps} shows the histogram of the values of the chi-square statistics for  $1,000,000$ steps in the chain, including $50,000$ burn-in steps), obtained from Sampson's monastery study \cite{Sampson68}. The vertical line denotes the value of the chi-square statistic for the observed monk dataset, indicating a large $p$-value of $0.986$  and thus a pretty good model fit.   
A similar histogram in Figure \ref{fig:Bay-1MgofsHistogram} shows that the $p_1$ model does not fit the Chesapeake Bay food web data so well: the estimated $p$-value is $0.03459$ after $1,000,000$ moves.
\begin{figure}
\centering
\begin{subfigure}{.45\textwidth}
  \centering
  \includegraphics[width=1\linewidth]{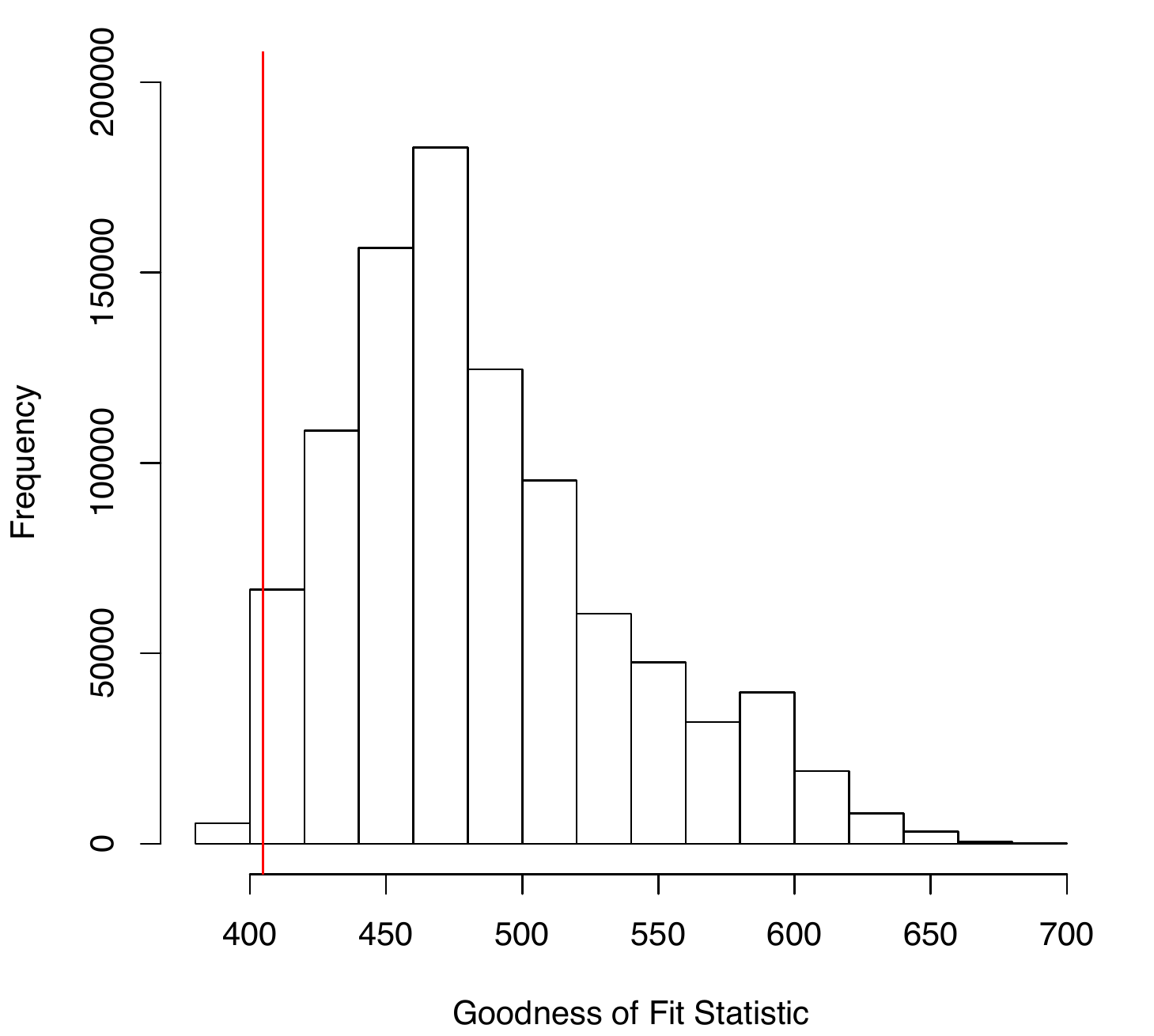}
  \caption{Affinity network derived from Sampson's monastery data set in \cite{Sampson68}.  Observed chi-square value: $404.7151$. $p=0.986$.}
	\label{fig:Monk-1MgofsHistogramExcl50kBurninSteps}
\end{subfigure}
\quad
\begin{subfigure}{.45\textwidth}
  \centering
  \includegraphics[width=1\linewidth]{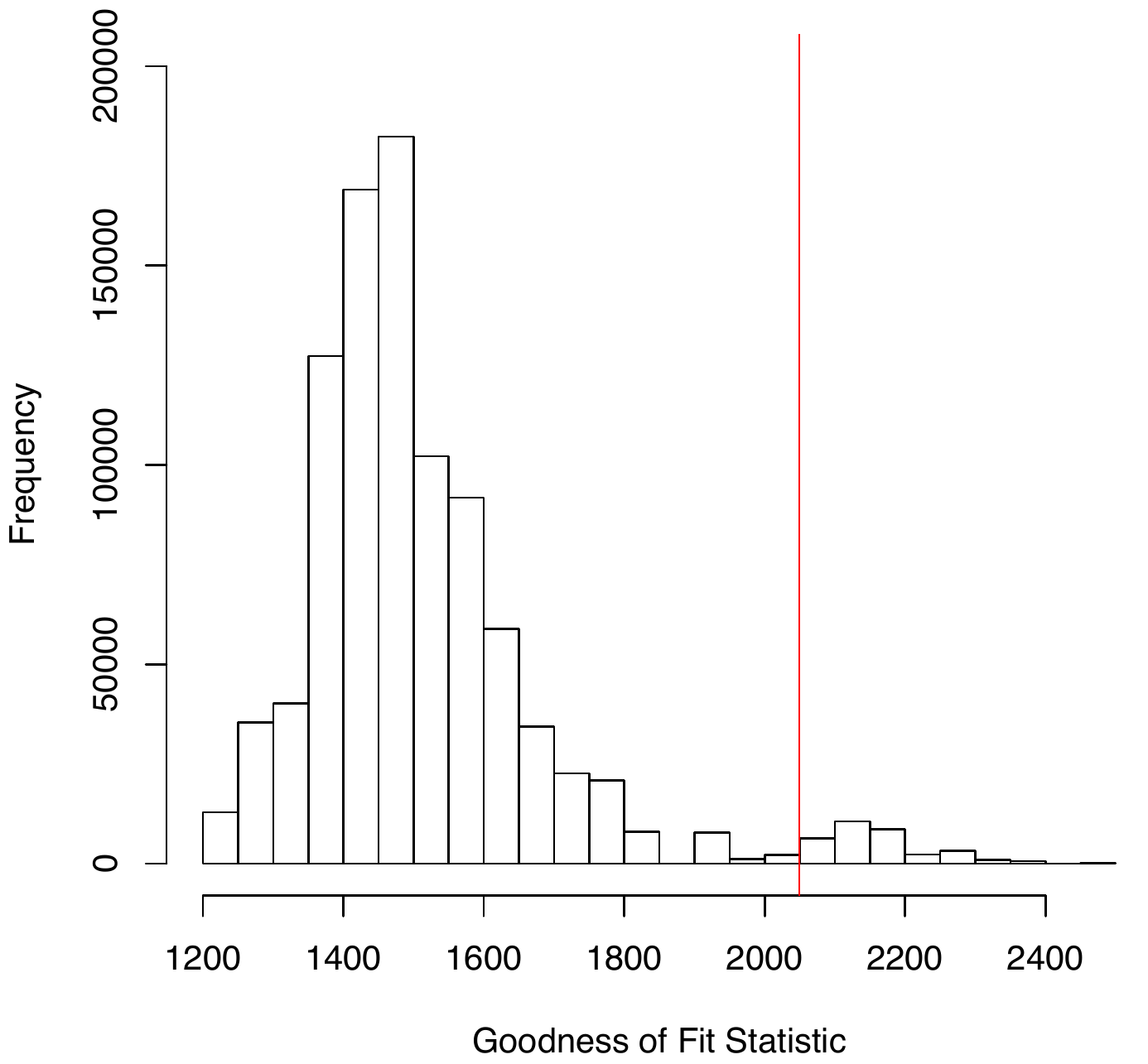}
  \caption{Chesapeake food web data set derived from \cite{BayData}. Observed chi-square value: $2049.403$. $p=0.03459$.}
	\label{fig:Bay-1MgofsHistogram}
\end{subfigure}
\caption{Sampling distribution of the chi-square statistic: histograms  from simulation running Algorithm~\ref{alg:wrapper} for the $p_1$ model with edge-dependent reciprocation.}
\end{figure}

\medskip
This paper is organized as follows. Section~\ref{sec:hypergraph} develops the combinatorial approach to the construction of Markov bases dynamically, and provides the necessary mathematical background. 
Section~\ref{sec:p1}  illustrates the developed methodology for the Holland and Leinhardt's $p_1$ model. Examples and simulations are in Section~\ref{sec:examples}. 
Specifically, further discussion and analyses of the model fit for the directed networks arising from the monk and food web data can be found in Sections~\ref{sec:ex:bay} and~\ref{sec:monk}. 
Sections~\ref{sec:kenya} and~\ref{sec:hl} provide studies of mobile money networks of a Kenyan family, and of four networks simulated from the $p_1$ distribution, respectively. 
Finally, simulations on a small synthetic network in \ref{sec:OHT} indicate good mixing times, and quick convergence of the $p$-value estimate (e.g., see Figure~\ref{fig:OHTpvalues-500k}).   As this is best illustrated when the entire fiber has been determined exactly, we also consider a small $591$-network fiber for an undirected graph on $8$ nodes from Section 5.1. in \cite{OHT}.  Our walk  explores the entire fiber in as little as $15,000$ moves and the total variation distance from the uniform distribution is below 0.25 after $10,000$ moves. This could be due to the fact that the steps in the simulated walks are longer than minimal Markov moves would suggest, since we are generating a superset of the Graver basis in our algorithm.

\section{Parameter hypergraph of a log-linear model: revised Metropolis-Hastings}
 \label{sec:hypergraph}

Markov and Graver bases arise as combinatorial signatures of log-linear models, and this natural correspondence is rooted in the algebra-geometry dictionary. In this section, we briefly describe the mathematical construction that allows us to dynamically generate applicable moves for sampling fibers of general log-linear models. 

\subsection{Markov bases: fundamentals} 
Consider a log-linear model on $m_1\times\cdots m_r$ a contingency table $U$ with sufficient statistics vector $S$. 
Let $u \in \Z_{\geq 0}^{m_1 \times \cdots \times m_r}$ be a realization of the  table $U$, with $S(u)=s(u)$. The fiber of $u$, which we will denote  $\fiberofu\subset \Z_{\geq 0}^{m_1 \times \cdots \times m_r}$
(or simply $\fiber$ if the observed table $u$ is implied from the context), 
 is the space of all realizations $v$ of the table whose sufficient statistics are the same as that of $u$; i.e., $s(v)=s(u)$.  For two tables in the same fiber $u, v \in \fiber$, the entrywise difference $u-v$ is called the \emph{move} from table $v$ to table $u$.  This move $u-v$ is another $r$-way  table with entries equal to zero in the cell $(i_1,\dots,i_r)$ if $u_{i_1,\dots,i_r}=v_{i_1,\dots,i_r}$, a positive integer in the $(i_1,\dots,i_r)$ cell if $u_{i_1,\dots,i_r} > v_{i_1,\dots,i_r}$, and a negative integer in the $(i_1,\dots,i_r)$-cell otherwise. Note that, by definition, $S$ is linear, thus  the sufficient statistic of any move connecting two tables in the same fiber, $S(u-v)$, is zero. In particular, adding a move to a contingency table does not change the values of the sufficient statistics vector. We will call any table $m \in \Z^{m_1 \times \cdots \times m_r}$ such that $S(m)=0$  a \emph{Markov move} on $\fiber$.  Thus, to discuss walks on  a fiber, we may either specify the start and target tables $v$ and $u$, or the Markov move $m=u-v$.

A \emph{Markov basis} $B$ is a set of Markov moves such that for {any}  fiber $\fiber$ and any two contingency tables $u,v\in\fiber=\fiberofu$, there exists a sequence of moves $m_1,\dots,m_k\in B$ such that $v$ is reachable from $u$ by the corresponding walk on the fiber $\fiberofu$, i.e., $u=v+\sum_{i=1}^k m_i$ and each partial sum  $u_l=v+\sum_{i=1}^l m_i$, $l<k$, is a table in the fiber $\fiberofu$ (that is, $u_l$ has nonnegative entries). 
The existence and finiteness of a Markov basis guaranteed by the fundamental theorem of Markov bases \cite{DS98}, which states that the moves correspond to generators of an algebraic object (namely, the toric ideal) associated with each log-linear model. 
Equipped with a set of moves,  one can perform a random walk on the fiber $\fiber$. A priori, the resulting Markov chain need not be irreducible;  however, if the set of moves is a {Markov basis}, then irreducibility is guaranteed. Moreover, a Metropolis-Hastings algorithm can be used to adjust the transition probabilities, returning a chain whose stationary distribution  
is exactly the conditional distribution on the given fiber. 

In this section, we discuss how to dynamically construct arbitrary elements of a  Markov basis $B$ for any log-linear model using \emph{the parameter hypergraph} of the model. For simplicity, we restrict ourselves to log-linear models with $0/1$ design matrices (that is, parameters do not appear with multiplicities in the model parametrization), although the definition and construction could be extended to  a more general case.  As mentioned in the introduction, this will be specifically useful in several cases: when $B$ cannot be computed in its entirety, e.g. when the model is not decomposable,   so that the divide-and-conquer strategy of \cite{DS04} cannot be applied, and  when sufficient statistics of the model are more complex than  table marginals. To that end, we define the main tool of our construction. 

\subsection{From tables to hypergraphs} 

Let ${\modelM}:=\modelM_S$ be any log-linear model for discrete random variables $Z_1,\dots,Z_m$ with sufficient statistics $S$. Suppose that the joint probabilities of the model are such that the parameters $\theta_1,\dots,\theta_n$ appear without multiplicities (that is, $S$ can be obtained from the table in a linear fashion). 
\begin{defn} 
The model ${\modelM}$ is encoded by a hypergraph $H_{\modelM}$ on the vertex set $\theta_1,\dots,\theta_n$, which is constructed as follows:  
$\{\theta_j\}_{j\in J}$ is an edge of $H_{\modelM}$ if and only if the index set $J$ describes one of the joint probabilities in the model; that is, there exist values $i_1,\dots,i_m$ such that, up to the normalizing constant, $Prob(Z_1=i_1,\dots,Z_m=i_m)\propto\prod_{j\in J} \theta_j$. 
  The hypergraph $H_{\modelM}$ is called the \emph{parameter hypergraph of the model $\modelM$}. 
\end{defn}
\begin{notation}
For convenience let us gather here the notational conventions we will use throughout. 
Log-linear models will be denoted by $\modelM_S$ with sufficient statistics $S$, or simply $\modelM$ when $S$ clear from context. 
The parameter hypergraph $H_{\modelM}=(V,E)$ has vertex set $V$ and edge set $E$.  
Edges in the hypergraph are written as products of parameters instead of the usual lists, e.g., $\theta_1\cdots\theta_k$ will represent the edge  $\{\theta_1,\dots,\theta_k\}$. 
\end{notation}
The easiest way to understand $H_{\modelM}$ is by viewing it as depicting the structure of parameter interactions. Since vertices of the hypergraph represent  parameters of the model, edges in $H_{\modelM}$ collect all the parameters that appear in a joint probability under the model. There is a one-to-one map between the contingency table cell labels and edges in the parameter hypergraph. 
Let us illustrate on two simple but familiar examples. 

\begin{example}[Two independent random variables]
\label{ex:independenceHypergraph}
Consider the model of independence of two discrete random variables $Z_1$ and $Z_2$, taking $a$ and $b$ values, respectively.
Denote the marginal probabilities $Prob(Z_1=i)$ and $Prob(Z_2=j)$ by $x_i$ and $y_j$, respectively. 
 Since the {independence model} for $Z_1$ and $Z_2$ is specified by the formula $P_{ij}:=Prob(Z_1=i,Z_2=j) = x_iy_j$, we see that the parameter hypergraph $H_{Z_1\ci Z_2}$ has $a+b$ vertices: $x_1,\dots,x_a,y_1,\dots,y_b$ and an edge between every $x_i$ and $y_j$. Thus, in this case, the hypergraph is just a complete bipartite graph on $\{x_1,\dots,x_a\}\sqcup \{y_1,\dots,y_b\}$, depicted in Figure~\ref{fig:indep}. 
\end{example}

\begin{example}[Quasi complete independence]
\label{ex:quasiIndependenceHypergraph}
For a $l \times m \times n$ table, the quasi complete independence model is a complete independence model with structural zeros. If the cell $(i,j,k)$ is a structural zero, then $Prob(Z_1=i,Z_2=j, Z_3=k)=0$, otherwise $Prob(Z_1=i,Z_2=j, Z_3=k) = x_iy_jz_k$ where $x_i=Prob(Z_1=i)$, $y_j=Prob(Z_2=j)$, and $z_k=Prob(Z_3=k)$ are marginal probabilities.   

 To obtain the parameter hypergraph for the quasi complete independence model, we start with the complete 3-partite hypergraph with vertex partition $V_1$, $V_2$, and $V_3$ such that $\#V_1=l$, $\#V_2=m$, and $\#V_3=n$, then remove every edge that corresponds to a cell with a structural zero. The hypergraph in Figure~\ref{fig:quasiInd} is the parameter hypergraph for the quasi complete independence model on a $3 \times 3 \times 3$ table where all cells are structural zeros except $(1,1,1), (1,1,2), (2,2,2), (2,3,3), (3,2,1)$, and $(3,3,3)$.
\end{example}

\begin{figure}
\centering
\begin{subfigure}{.45\textwidth}
  \centering
	\includegraphics{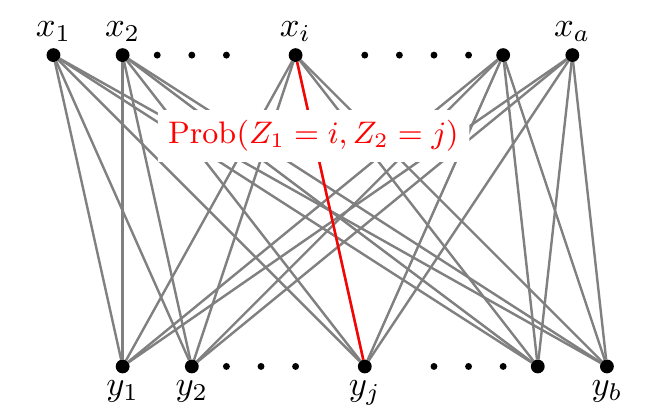}
	\caption{	Independence model: Example~\ref{ex:independenceHypergraph}.
			}
	\label{fig:indep}  
\end{subfigure}%
\begin{subfigure}{.45\textwidth}
  \centering 
	\includegraphics[width=1.5in]{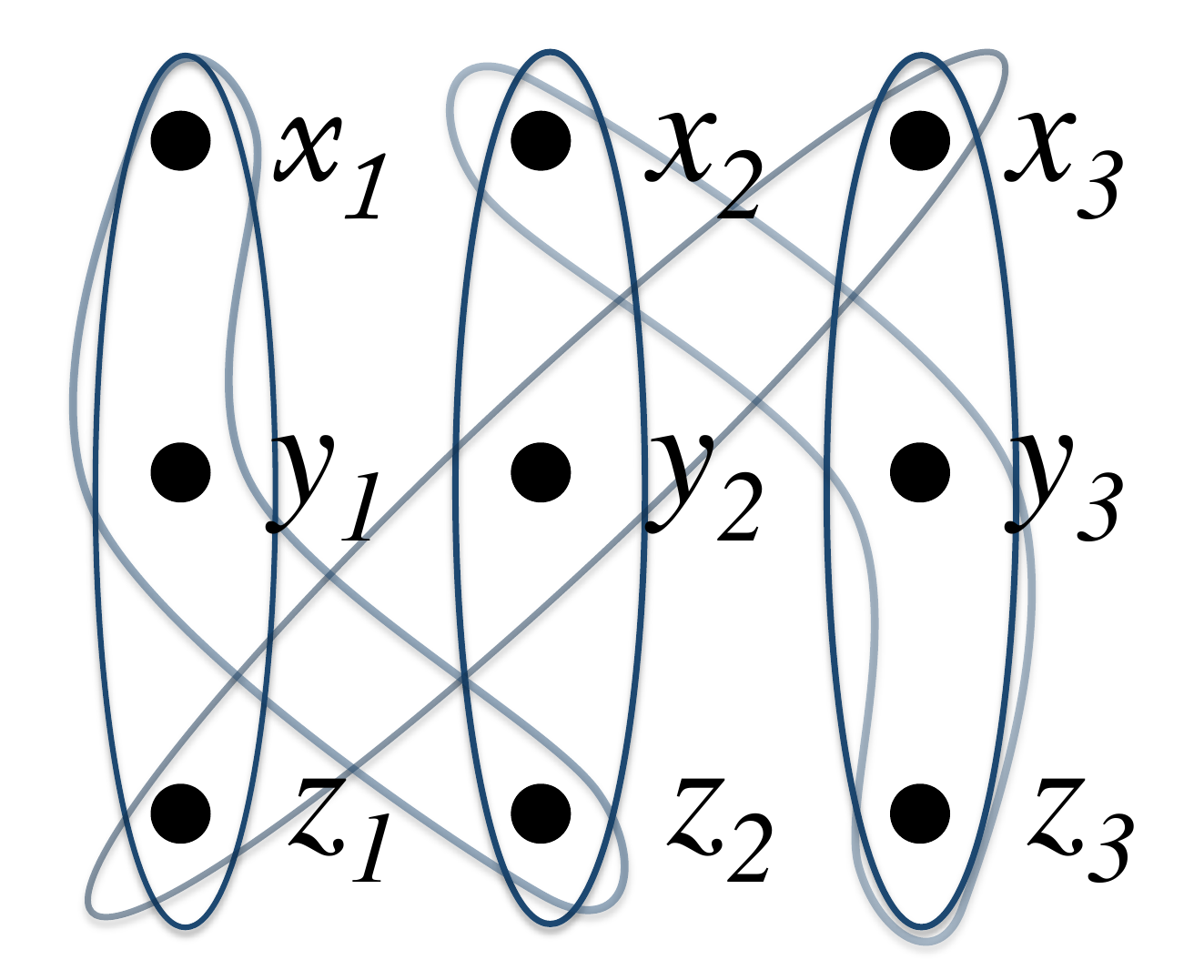}
	\caption{Quasi independence model: Example~\ref{ex:quasiIndependenceHypergraph}.}
	\vspace{-1cm}
	\label{fig:quasiInd}  
\end{subfigure}
\caption{Two examples of parameter hypergraphs}
\end{figure}

In the next section (Definition~\ref{defn:p1hypergraph}) we will see a more complex example in $\Hp$, the parameter hypergraph for the version of the $p_1$ model that assumes edge-dependent reciprocation.

\smallskip
A crucial observation about the parameter hypergraph is that it not only encodes the parameter interactions, but \emph{any observed table} can be viewed as a subset of its edges,  with multiplicities if the model allows them.  Specifically, suppose the table $u$ has an entry $1$ in the cell $(i_1,\dots,i_r)$. If the model postulates $Prob(X_1=i_1,\dots,X_r=i_r)\propto \theta_{j_1}\cdots \theta_{j_k}$, then the $(i_1,\dots,i_r)$-cell entry is represented by the edge $\theta_{j_1}\cdots \theta_{j_k}$.  A larger entry (say, $3$) in the table would be represented by an edge with multiplicities (the edge   $\theta_{j_1}\cdots \theta_{j_k}$ would have multiplicity $3$). Multiplicities are recorded with a function $\mu: E \to \Z$ (e.g. $\mu(\theta_{j_1}\cdots \theta_{j_k})=3)$.

\begin{defn}\label{defn:EdgesForTable}
	The list of edges 
	\begin{align*}
		\{\theta_{j_1}\cdots \theta_{j_k}:  u_{(i_1,\dots,i_r)}>0  \textit{ and } Prob(X_1=i_1,\dots,X_r=i_r)\propto \theta_{j_1}\cdots \theta_{j_k}\}, \\ 
		\text{ where  edge } \theta_{j_1}\cdots \theta_{j_k}\text{ appears }\mu(\theta_{j_1}\cdots \theta_{j_k})=u_{(i_1,\dots,i_r)} \text{ times,} 
	\end{align*}
	 will be denoted by $e(u)$.  It is the multiset of edges	representing the table $u$, and has   support in the edge set $E$ of the parameter hypergraph. 
\end{defn}

Next, notice that sufficient statistics $S(u)$ can be calculated from the hypergraph edges $e(u)$, since the vertices covered by $e(u)$ represent those natural parameters that affect the computation of $S(u)$. In the independence model example (cf. Example~\ref{ex:independenceHypergraph}), if $u$ is the $2\times 2$ table with $1$ in cell $(1,2)$ and a $2$ in the cell $(2,1)$, then $e(u) = \{x_1y_2, x_2y_1, x_2y_1\}$. The sufficient statistics of the table under $Z_1\ci Z_2$ are the row and column sums; the first row having sum $1$ means that $x_1$ appears once in the set of edges $e(u)$; in other words, the degree of the vertex $x_1$ is  $1$. The first column having sum $2$ means that $y_1$ has degree $2$ in  $e(u)$. 
Therefore, the vector of sufficient statistics $s(u)$ equals the \emph{degree vector} of the multi-hypergraph $(V,e(u))$.  It is obtained by simply counting the number of edges incident to each vertex in $e(u)$ and setting the degree of all other vertices in $V$ to zero. 

Finally, we describe how to construct and explore the fiber $\fiberofu$.  Preserving the value of the vector $s(u)$ means finding another edge set $e(v)$ such that the degree vector of $e(v)$ is the same as that of $e(u)$. If we view the edges $e(u)$ as colored red and $e(v)$ blue, then the move $v-u$ corresponds to a  collection of edges $\left(e(u),e(v)\right)$, where each vertex appears in the same number of blue and red edges.  

We have thus shown the following is an equivalent way to view the fiber $\fiberofu$ and its connecting moves. 
\begin{theorem} \label{thm:FiberAndMovesOnHypergraph}
	Recall that an observed table $u$ is represented by a multiset $e(u)$ of edges on the hypergraph $H_{\modelM}$.  
	\begin{enumerate}
		\item[(a)] 
			The fiber $\fiberofu$ consists of all multisets of edges of $H_{\modelM}$ with degree vector equal to $s(u)$. 
		\item[(b)]
		Any move $v-u$ in the Markov  basis connecting  $u$ to some  $v\in\fiberofu$ is represented by the edge sets $\left(e(u),e(v)\right)$ over the  parameter hypergraph $H_{\modelM}$   such that the degree vector of $e(v)$ is the same as that of $e(u)$.
	\end{enumerate}
\end{theorem}
We call such a collection $\left(e(u),e(v)\right)$ a \emph{(color-)balanced} edge set; it was defined and discussed in more detail in \cite{PS13}, where it was shown that such sets constitute a Markov (and in fact, the  Graver) basis for any parameter hypergraph.  Complexity of minimal Markov moves to connect the given (unrestricted) fiber was studied in \cite{GP13}.  
For convenience, let us summarize here the hypergraph notation we will use in the following section.  
\begin{notation} 
\label{notationRedBlue}
For an observed table $u$, the set of red (observed) hyperedges $e(u)$ will be denoted by $\cR$, and any blue set that balances the vertices covered by $\cR$ will be denoted by $\cB$.  Note that every $\cB$ corresponds to a table $v\in\fiberofu$.  The move $v-u$ will be denoted as $\mathcal W=(\cR, \cB)$.  
\end{notation}
\begin{rmk}
	By abuse of notation, we will also denote by $(\cR,\cB)$ only those edges over $H_{\modelM}$ representing the non-zero entries of the move $v-u$. Indeed, if a cell has the same value in both tables, the move directly connecting the tables does not affect that cell, thus the corresponding edge need not be recorded in $(\cR,\cB)$. If it is included in this set, then the move simply subtracts and adds $1$ to the cell in the table, that is, it removes and then adds back the particular edge in $e(u)$. 
\end{rmk}

\subsection{Sampling constraints and restricted fibers} 
As mentioned briefly in the introduction, 
a Markov basis will connect all 
 table realizations
   in a  fiber that are subject to the constraint that each table entry  is non-negative.  However, in the presence of table cell bounds or structural zeros in the model (e.g. \cite[\S 5.1]{BFH75}),  Markov moves will inevitably produce tables whose cell entries are too large, even if they satisfy the sufficient statistics (say, the realization of the table reached by the random walk  will have the given marginals, but some cells will be out of bounds). These sampling constraints  often arise in real-world data. In the network modeling case, a structural zero means a certain relation or edge can never be observed, while a cell bound puts a restriction on how many times an edge between two nodes can be observed  in any instance of the network.  In fact, most (simple) network models begin with a basic assumption that allows only one edge per dyad, for example,  the $p_1$ model \cite{HL81} (see also \cite{FPRholland})  and the beta model \cite{CDS11}. 
This clearly introduces another problem for running random walks on fibers: at any given step, the table or network produced may not be observable, and so many of the steps in the walk will be rejected. In fact, the rejection is likely to occur because the usual Markov bases are blind to data and sampling constraints. To compound this problem, a Markov basis only guarantees that the fiber of non-negative  table realizations is connected. It is quite reasonable to expect  that there exist two of them that can be connected only by a walk that passes through another table realization which does not satisfy the additional cell bounds. 
In this sense, the sampling constraints have suddenly \emph{disconnected} the fiber $\fiber$!  
 With this in mind,  we will differentiate between the usual fiber $\fiber$ and what we call the observable fiber $\simplefiber$: 
\begin{defn}
	 The \emph{observable fiber} $\simplefiber\subsetneq\fiber$ is the set of all realizations $u$  of the contingency table $U\in \Z_{\geq 0}^{m_1 \times \cdots \times m_r}$  with nonnegative entries and  sufficient statistic $S=s$  that respect the sampling constraints of the model, i.e. integer bounds on cells or structural zeros. 
\end{defn}
For example, in the $p_1$ model, the observable fiber $\simplefiber$ contains only \emph{simple} directed graphs, 
  which means each cell in the contingency table representing the directed graph is either a $0$ or a $1$. 
  Naturally, there is a corresponding condition on the hypergraph: no edge in $e(u)$ representing the table $u$  can have multiplicity larger than $1$. 
Thus any move $(\cR,\cB)$ applied to $e(u)=\cR$ must be such that in the resulting set of edges, $\left(e(u)\setminus \cR\right) \cup \cB\subseteq H_{\modelM}$, every edge appears at most once.  

Thus, a natural question arises: does there exist a finite set of moves that connects the observable fiber? The answer is known in the literature under the name of Graver basis or distance-reducing moves. 
 Hara and Takemura study the observable fibers for $0/1$ contingency tables, that is, tables with cell bound of $1$ everywhere,  and show \cite[Proposition 2.1]{HT}  that the squarefree part of the Graver basis will connect \emph{any} fiber $\simplefiber$ \emph{respecting} $0/1$ sampling constraints.  Here, ``squarefree part" simply means that each entry in the table representing the move $u-v$ is either $0$ or $1$; we will say that such a move \emph{respects the $0/1$ sampling constraint}. Their result is, in fact, more general, and applies to higher integer cell bounds and structural zeros as well:  

\begin{proposition}[\cite{HT}]
\label{prop:Graverbasis} 
	The elements of the Graver basis which respect the sampling constraints suffice to connect the observable fiber in all cases where sampling constraints are integer bounds on cells. 
\end{proposition}
The proof relies on an algebraic fact that moves correspond to binomials in  a toric ideal, and every binomial arising from the given model can be written as what is called a conformal sum of Graver basis elements. We will not go into technical details of this result here; the reader is referred to \cite{St96} and recent text \cite{AHT2012}. 

\subsection{Applicable moves and revised Metropolis-Hastings} 

In general, the set of squarefree moves from the Graver basis is much larger than a minimal Markov basis. In particular, this set almost never equals the squarefree moves from a minimal basis. Moreover, it is notoriously difficult to compute, providing another reason against pre-computing the moves for the given model, and instead, generating dynamically only those moves that can be applied to the observed table or network and remain in the observable fiber $\simplefiber$. 
\begin{defn}
	A move $v-u$ is said to be \emph{applicable} to a point $u$ in the fiber (equivalently, to the network represented by a table $u$)  if it produces another point $v$ in the observable fiber $\simplefiber$, respecting the sampling constraints of the model at hand. 
	
	In terms of the hypergraph edges, the move $v-u$, represented as $(\cR,\cB)$, is applicable if $\left(e(u)\setminus \cR\right) \cup \cB = e(v)$ for some  table $v\in\simplefiber$. 
\end{defn}

 We can thus extend Theorem~\ref{thm:FiberAndMovesOnHypergraph} to characterize applicable Graver moves in terms of the parameter hypergraph: 
  By Theorem 2.8 in \cite{PS13} and the Fundamental Theorem of Markov bases, \emph{any} move corresponds to a balanced edge set of $\cH_{\modelM}$. Furthermore, moves in the Graver bases  correspond to the \emph{primitive} balanced edge sets of $\cH_{\M}$.  We can summarize  applicable Graver moves in terms of $\cH_{\modelM}$ in the following way. 
 \begin{cor}\label{cor:applicableGraver}
 Adopt Notation~\ref{notationRedBlue}. 
 Any move $v-u$ in the \emph{Graver}  basis that is applicable to $u$ is a set of edges $\left(\cR, \cB\right)$ such that:
 \begin{enumerate}
 \item  $\cR\subseteq e(u)$, 
 \item  $\left(e(u)\setminus \cR\right) \cup \cB = e(v)$ for some  table $v\in\simplefiber$, and 
 \item  
there exists no move $\left(\cR', \cB'\right) $ such that $\cR'\subset\cR$ and $\cB'\subset\cB$. 
\end{enumerate}
 \end{cor} 
In the result above,  1.\ ensures non-negativity of the resulting table $v$, 2.\ ensures the move is applicable, and 3.\ ensures the move is a Graver basis element. In practice, however, checking 3. is a non-trivial task; instead, an algorithm with positive probability for producing each Graver move suffices for goodness of fit testing  purposes. Thus,  in Section~\ref{sec:p1} we run walks on fibers using elements of the Graver basis along with some larger applicable moves as well.

\smallskip

The remainder of this section discusses how to construct applicable moves and embeds the combinatorial idea from Corollary~\ref{cor:applicableGraver} within the Metropolis-Hastings algorithm to perform random walks on fibers. 

\smallskip

{\footnotesize  
\begin{algorithm}[H]
\label{alg:MH}
\LinesNumbered
\DontPrintSemicolon
\SetAlgoLined
\SetKwInOut{Input}{input}
\SetKwInOut{Output}{output}
\Input{
$u\in \mathcal T(n)$, a contingency table (or $G=g$, a network represented by $u$), \\
\quad $S(u)=s(u)$, the sufficient statistics for the model $\modelM$,\\ 
\quad $N$ the number of steps,\\
\quad $f\left(\cdot|S(u)\right)$ conditional probability distribution,\\
\quad $GF(\cdot)$, test statistic}
\Output{Estimate of $p$-value}
\BlankLine
Compute the MLE $\tilde{p}$.\;
Set $GF_{\text{observed}}:=GF(u)$.\;  
Randomly select a multiset of hyperedges $\cR$ from $e(u)$.\; 
Find a multiset of hyperedges $\cB$ from $\cH_{\M}$ that balances $\cR$, ensuring that each Graver move $(\cR, \cB)$ 
has  positive probability of being constructed.\;
Set $m=\cR-\cB$.\; 
$q=\min\left\{1, \frac{f(U=u+m|t)}{f(U=u|t)}\right\}$\;
$u = \begin{cases} 
u+m, & \text{ with probability } q\\
u, & \text{ with probability } 1-q\end{cases}$\;
\uIf {$GF\left(u\right)>GF_{\text{observed}}$}{   
	$k=k+1$.}
Repeat Steps 3-9 $N$ times.\;
Output $\frac{k}{N}$.
\caption{Revised Metropolis-Hastings}
\end{algorithm}
}

If the procedure for finding $\cB$ in step 4 is symmetric and non-periodic, then Algorithm \ref{alg:MH} is a Metropolis-Hastings algorithm and as $N \to \infty$ the output will converge to $P(GF(U) \geq GF(u) \ | U \in \fiberofu)$ (\cite{DSS09}, \cite{RC99}).
Ideally, Step 4 should take advantage of the specific structure of the hypergraph. 
 For example, we employ this process to implement Algorithm \ref{alg:MH} and produce applicable moves on the fly for the Holland-Leinhardt $p_1$ model in Section~\ref{sec:p1}. 

The use of  $H_{\modelM}$ allows us to bypass two crucial issues of the usual chain, as stated in \cite{DS98},  which relies on precomputing a minimal Markov basis, and which are summarized in the last paragraph of \cite{Dobra2012}. First, Algorithm \ref{alg:MH} does not require  computing the full Markov basis, or the full Graver basis as may be required due to sampling constraints.  Second, the rejection step from the usual Metropolis-Hastings is bypassed, since rejections are due to the fact that most moves drawn from the full Markov basis will be non-applicable to the current table. This, in turn, should have positive impact to the mixing time of the chain.

\begin{example} \label{ex:independence-model-move}
Suppose we observe a $5\times5$ contingency table all of whose entries are $0$ except the $(1,1)$ and $(2,2)$ entries, which are $1$. There are 200 moves in a minimal Markov basis for the independence model $Z_1\ci Z_2$.  However, only one of those is applicable: namely
$$\begin{array}{|c|c|c|c|c|}\hline -1 & 1 & 0 & 0 & 0 \\\hline 1 & -1 & 0 & 0 & 0 \\\hline 0 & 0 & 0 & 0 & 0 \\\hline 0 & 0 & 0 & 0 & 0 \\\hline 0 & 0 & 0 & 0 & 0 \\\hline \end{array}\ ,$$
or, written in terms of the parameter hypergraph, $\mathcal W=( \cB, \cR)$ where $\cB=\{ x_1y_2, x_2y_1\}$ and $\cR=\{x_1y_1, x_2y_2\}$.
 This move replaces the entries $(1,1)$ and $(2,2)$ by $0$, and entries $(1,2)$ and $(2,1)$ by $1$.  Any other move will produce negative entries in the table and thus move outside the fiber.  A more interesting example can be similarly constructed on a  $k$-way table that is either sparse or has many non-zero entries but $\simplefiber$ allows  only $0/1$ entries.

Next, suppose the observed table is 
$$u=\begin{array}{|c|c|c|c|c|}\hline 3 & 2 & 0 & 1 & 0 \\\hline 1 & 0 & 0 & 0 & 1 \\\hline 0 & 0 & 0 & 2 & 0 \\\hline 0 & 1 & 0 & 0 & 0 \\\hline 0 & 0 & 0 & 0 & 0 \\\hline \end{array}\ .$$ 
The table $u$ is represented by the multiset of edges 
\[
e(u)=\{x_1y_1,x_1y_1,x_1y_1,x_1y_2,x_1y_2, x_1x_4, x_2y_1,x_2y_5,x_3y_4,x_3y_4, x_4y_2 \} 
\]  
from the independence model (hyper)graph illustrated in Figure~\ref{fig:indep}. Denote the bipartite (hyper)graph in Figure~\ref{fig:indep} as $G$. It is known that any Markov move for the independence model corresponds to a collection of closed even walks on $G$, and any Graver move corresponds to a primitive closed even walk on $G$. For a detailed account of the correspondence between primitive balanced edge sets of $G$ and primitive closed even walks see \cite{Vill}. Due to this correspondence, a natural procedure for performing Step 4 in Algorithm~\ref{alg:MH} is to randomly select a set of edges from $e(u)$, say, $\cR=\{x_1y_1,x_2y_5, x_3y_4\}$, and then complete a closed even walk on $\cR$, so that the new edges form $\cB=\{x_2y_1, x_3y_5, x_1y_4\}$.  Notice $\cR$ and $\cB$ have the same degree vector and $(\cR, \cB)$ is applicable to $u$. This move is depicted in Figure~\ref{fig:independence-model-move}. The first figure is the parameter hypergraph. The second represents the observed table $e(u)$, with edges in $\cR$ highlighted. The third is the edge set $e(v)$ with $\cB$ highlighted. 
The resulting table is 
$$v=\begin{array}{|c|c|c|c|c|}\hline 2&2&0&2&0 \\\hline 2&0 & 0 & 0 & 0 \\\hline 0 & 0 & 0 & 1 & 1 \\\hline 0 & 1 & 0 & 0 & 0 \\\hline 0 & 0 & 0 & 0 & 0 \\\hline \end{array}\ .$$
\begin{figure}
	\label{fig:independence-model-move}
	\includegraphics{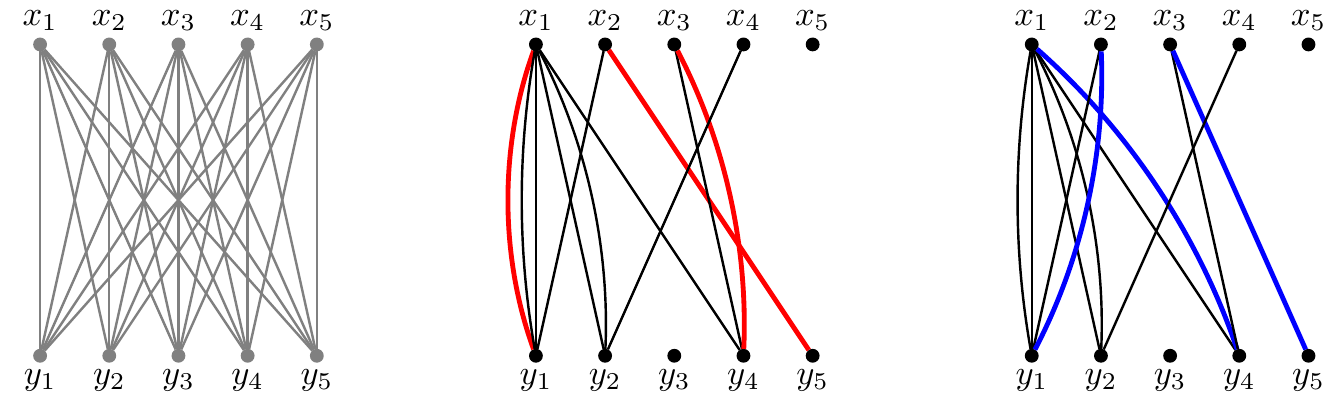}
	\caption{Example~\ref{ex:independence-model-move}: parameter hypergraph (left); observed tables $e(u)$ and $e(v)$ with applicable move $(\cR, \cB)$ highlighted (center and right). }
	\label{fig:independence-model-move}
\end{figure}
\end{example}

\section{Application to the $p_1$-model
}\label{sec:p1}

In a seminal 1981 paper \cite{HL81}, Holland and Leinhardt described what they referred to as the $p_1$ model for describing dyadic relational data in a social network
summarized in the form of a directed graph. Their model,  which is log-linear in form  (\cite{FW81}), allows   for effects due to differential attraction (popularity) and
expansiveness, as well as an additional effect due to reciprocation.   
For each dyad, a pair of nodes ($i,j)$, the parameter $\alpha_i$ describes the effect of an outgoing edge from $i$, and $\beta_j$ the effect of an incoming edge pointed towards $j$, while $\rho_{ij}$ corresponds to the added effect of reciprocated edges. The parameter $\theta$ quantifies the average ``density" of the network, i.e. the tendency of having edges, and $\lambda_{ij}$ is a normalizing constant to ensure that the probabilities for each dyad $(i,j)$ add to 1.

Given a directed graph, each dyad $(i,j)$ can occur in one of the four possible configurations: no edge, edge from $i$ to $j$, edge from $j$ to $i$, and a pair of reciprocated edges between $i$ and $j$. The model  postulates that, for each pair $(i,j)$, the probability of observing the four possible configurations, in that order, satisfy the following equations: \begin{align*}
    p_{ij}(0,0) &= \exp[\lambda _{ij}]\\
    p_{ij}(1,0) &= \exp[\lambda _{ij} + \alpha_i + \beta_j + \theta]\\
    p_{ij}(0,1) &= \exp[\lambda _{ij} + \alpha_j + \beta_i + \theta]\\
    p_{ij}(1,1) &= \exp[\lambda _{ij} + \alpha_i + \beta_j + \alpha_j + \beta_i + 2\theta + \rho_{ij}].
\end{align*}
where 
\[
\sum_i \alpha_i = \sum_j \beta_j=0.
\]
We will focus on the \emph{edge-dependent} version of the reciprocation parameter, where $\rho_{ij}=\rho_i+\rho_j+\rho$.

Making the following substitutions  
\[
 \alpha_i'=e^{\alpha_i+ \theta}, \quad  \beta_i'=e^{\beta_i}, \quad  \rho_{i}'=e^{\frac{1}{2}\rho + \rho_{i}} 
\]
and ignoring the superscripts for convenience, we arrive at the following simplified equations to describe the probability of observing each configuration for a pair $(i,j)$:

\begin{align*}
    p_{ij}(0,0) &= \lambda _{ij}  \\
    p_{ij}(1,0) &= \lambda _{ij}\alpha_i \beta_j  \\
    p_{ij}(0,1) &= \lambda _{ij} \alpha_j \beta_i  \\
    p_{ij}(1,1) &= \lambda _{ij}  \alpha_i \beta_j  \alpha_j  \beta_i  \rho_i \rho_j.
\end{align*}

While normalizing constants are usually ignored, we will follow \cite{PRF10} and treat $\lambda_{ij}$ as a model parameter.  The advantage of this technique is that, given an observable network $g$, these extra parameters ensure that the sampling constraint of a dyad (pair) $\{i,j\}$ being observed in one and only one state is satisfied for all networks in $\fiberofg$.

\begin{defn}[The parameter hypergraph of the $p_1$ model]
\label{defn:p1hypergraph}
	We will denote the parameter hypergraph of the $p_1$ model as $\Hp$. Recall that the hyperedges of $\Hp$ are determined by the parameters appearing in the joint probabilities of the model. Thus, for the $p_1$ model with edge reciprocation there are three types of hyperedges: singletons (corresponding to $p_{ij}(0,0)$ for each dyad $(i,j)$), hyperedges of size $3$ (corresponding to $p_{ij}(1,0)$ and $p_{ij}(0,1)$), and hyperedges of size $7$ (corresponding to $p_{ij}(1,1)$). 
	
	 More formally, \mbox{$\Hp=(V_p,E)$}, where $V_p=\{\alpha_i,\beta_i,\rho_i: 1\leq i\leq n\}\cup\{\lambda_{ij}: 1\leq i<j\leq n\}$, and $ E= E_1 \cup E_3 \cup E_7$, with $E_1= \{ \lambda_{ij} : 1\leq i<j\leq n\}$, $ E_3= \{ \alpha_i\beta_j\lambda_{ij}: 1\leq i\neq j\leq n\}$, and $ E_7= \{\alpha_i\alpha_j\beta_i\beta_j\rho_i\rho_j\lambda_{ij}:1\leq i<j\leq n\}$. 
\end{defn}

\subsection{Markov moves for the $p_1$ model}

Here we describe the form of a Markov move $\mathcal W=(\cB, \cR)$ for the $p_1$ model with edge-dependent reciprocation in terms of the parameter hypergraph $\cH_{p1}$ given in Definition~\ref{defn:p1hypergraph}. The moves can be described in terms of balanced edge sets  on a graph obtained by contracting hyperedges in $\Hp$.  Note that by definition,  balanced edge sets on graphs reduces to collections of closed even walks. 

Let $A_n$ be the undirected bipartite graph on $2n$ vertices with vertex set $$V(A_n)=\{\alpha_i \ | \ 1\leq i\leq n\}\cup\{\beta_i \ | \  1\leq i\leq n\}$$ and edge set $$E(A_n)= \{ \alpha_i \beta_j \ | \ 1\leq i\neq j\leq n\}.$$ Let $K_n$ be the undirected complete graph on the $n$ vertices $\{\rho_j \ | \ 1\leq j\leq n\}$.  The graphs $A_n$ and $K_n$  can be constructed from $\cH_{p1}$ as follows.
To construct $A_n$ from $\cH_{p1}$, simply consider all hyperedges of size $3$ in $\cH_{p_1}$. Each of these hyperedges has vertices $\alpha_j,\beta_k,\lambda_{j,k}$ for some $1\leq j\neq k\leq n$. The contracted edges $\alpha_j\beta_k$ (with $\lambda_{j,k}$ deleted) are precisely the edges in $A_n$. 
	To construct $K_n$, consider all hyperedges of size $7$ in $\cH_{p1}$. Note that  each of these edges corresponds to an edge in $\Hp$ that has vertices $\alpha_j,\alpha_k,\beta_j,\beta_k,\lambda_{j,k},\rho_j,\rho_k$ for some $1\leq j\neq k\leq n$. Deleting all the vertices except $\rho_j,\rho_k$ from each hyperedge of size $7$ contracts them to size $2$, and the result is the complete graph on the $n$ vertices $\rho_j, j=1,\dots,n$.  

Let $\Hp |_{(3,7)}$ be the subhypergraph of $\Hp$ where $\mathcal V(\Hp |_{(3,7)}) = \mathcal V(\Hp)$ and $ E (\Hp |_{(3,7)})=\{ e \in  E( \Hp) \ | \ \#e=3 \text{ or } \# e=7\}$.  The previous two paragraphs describe a bijection between the edge sets of $A_n \cup K_n$ and $\Hp |_{(3,7)}$:
\begin{align*}
\phi: E(A_n \cup K_n)  & \to  E( \Hp |_{(3,7)})\\
\alpha_i \beta_j & \mapsto \alpha_i \beta_j \lambda_{ij} \\
\rho_i \rho_j  & \mapsto \alpha_i  \alpha_j \beta_i\beta_j \lambda_{ij}\rho_j \rho_k. \\
\end{align*}
\indent For a simple balanced edge set $W=(B,R)$ of $A_n \cup K_n$, the set $(\phi(B), \phi(R))$ may not be balanced.  However, it can become balanced by appending edges of the form $\{ \lambda_{ij}\}$ to the sets $\phi(R)$ and $\phi(B)$.  Thus, we define a lifting operation that grows
 $W$ to a simple balanced edge set of $\Hp$ in this manner:
\begin{align*}
\text{lift} W:&=(\cB, \cR)\text{, where }\\ 
\cB &= \phi (B) \cup \{ \lambda_{ij} \ | \deg_{\phi(R)}(\lambda_{ij}) > \deg_{\phi(B)} (\lambda_{ij})\}\text{ and }\\
\cR &= \phi (R) \cup \{ \lambda_{ij}  \ | \deg_{\phi(B)}(\lambda_{ij}) > \deg_{\phi(R)} (\lambda_{ij}) \}.
\end{align*}

Let $\Hp |_{(7)}$ be the subhypergraph of $\Hp$ that contains all the hyperedges of $\Hp$ of size 7. Let $\Hp |_{(3)}$ be the subhypergraph of $\Hp$ that contains all the hyperedges of $\Hp$ of size 3. If $\mathcal W=(\cB, \cR)$ is a balanced edge set of $\Hp$ then each $\rho_i$ in the hyperedges of size 7 of $\mathcal W$ must be color-balanced.  This implies that the $\alpha$'s and $\beta$'s are color-balanced with respect to $\Hp |_{(7)}$.  Thus, it follows that the $\alpha$'s and the $\beta$'s are color-balanced in $\Hp|_{(3)}$.  These observations are noted in \cite{PRF10}, but in algebraic terms using the binomials of the ideal of the hypergraph $I_{\Hp}$. 

Since a balanced edge set $\mathcal W=(\cB, \cR)$ on $\Hp$ is a move between two observable networks only if \mbox{$\deg_{\mathcal R}(\lambda_{ij})=\deg_{\mathcal B}(\lambda_{ij})\in\{0,1\}$}, we arrive at the following proposition.

\begin{prop}\label{prop:lift} A move between two observable networks $g_1$ and $g_2$ in the same fiber is of the form \emph{lift}$W$ such that $W$ is a balanced edge set on $A_n \cup K_n$ and  \mbox{$\deg_{\mathcal R}(\lambda_{ij})=\deg_{\mathcal B}(\lambda_{ij})\in\{0,1\}$.}
\end{prop}

\begin{corollary}\label{cor:lift}  For the $p_1$ model with edge-dependent reciprocation, the set of all $\mathcal W=(\cB, \cR)$ such that $\mathcal W=\emph{lift}(W)$ and $W$ is a balanced edge set of $A_n \cup K_n$ and \mbox{$\deg_{\mathcal R}(\lambda_{ij})=\deg_{\mathcal B}(\lambda_{ij})\in\{0,1\}$} connects the observable fiber $\simplefiber$ for every possible sufficient statistic $s$. 
\end{corollary}

\begin{rmk} The set of moves described in Corollary~\ref{cor:lift} is a superset of the square-free Graver basis.
\end{rmk}

\subsection{ Generating an applicable move}
Now that we have described the general form for the Markov moves for the $p_1$ model, we give an algorithm for generating an applicable move. Let $g=g_u \cup g_d$ be an observable network written as the union of its reciprocated part $g_u$ and its unreciprocated part $g_d$.
For a directed graph $G=(V,E)$, let undir($G$) be the edges of the skeleton of $G$ and let recip($G$)=($V$, recip($E))$ where recip$(E)=\{ (e_1, e_2)  \ : (e_1, e_2) \in E \text{ or } (e_2, e_1) \in E \}$.
 The following is a general algorithm for generating applicable moves for the $p_1$ model with edge-dependent reciprocation. It uses the fact that every balanced edge set of a graph corresponds to a set of closed even walks on that graph. The output is either an element of the Graver basis, or an applicable combination of several Graver moves, which themselves need not be applicable. Since the hyperedges of a balanced edge set on $\Hp$ each correspond to a dyadic configuration realizable in the network, we will return moves in the form $(b,r)$ where $b$ are the edges to be removed from the network and $r$ are the edges to be added.

{\footnotesize 
\begin{algorithm}[H]
\label{alg:wrapper}
\LinesNumbered
\DontPrintSemicolon
\SetAlgoLined
\SetKwInOut{Input}{input}
\SetKwInOut{Output}{output}
\Input{
$g=g_u \cup g_d$, a directed graph,\\
\quad $c_1$, the probability of choosing 1,\\
\quad $c_2$, the probability of choosing 2,\\
\quad $c_3$, the probability of choosing 3.\\
}
\Output{$(b,r)$, an applicable  move.}
\BlankLine
Generate $c$, a random number between 1 and 3 chosen with probabilities $(c_1, c_2, c_3)$ (weighted coin) 
\;
\uIf {$c=1$}{Use Algorithm~\ref{alg:type1move} to select a Type 1 move. Only reciprocated edges are removed and added in pairs.  A move of this type corresponds to a set of closed even walks on $K_n$}
\uIf {$c=2$}{Use Algorithm ~\ref{alg:type2move} to select a Type 2 move. Only unreciprocated edges are removed and added.  A move of this type corresponds to a set of closed even walks on $A_n$.}
\uIf{$ c=3$}{ Use Algortithm ~\ref{alg:type3move} to select a Type 3 move. Both types of edges are removed and added.  A move of this type corresponds to a set of closed even walks on $A_n$ and a set of closed even walks on $K_n$.}
\caption{Generating applicable moves for the $p_1$ model.}
\end{algorithm}
}

\smallskip
{\footnotesize  
\begin{algorithm}[H]
\label{alg:type1move}
\LinesNumbered
\DontPrintSemicolon
\SetAlgoLined
\SetKwInOut{Input}{input}
\SetKwInOut{Output}{output}
\Input{
$g_u$, a directed graph}
\Output{$(b,r)$, a Type 1 (reciprocated-only) applicable move.}
\BlankLine
Choose a random subset $r_0$ of edges from undir$(g_u)$.\;
\For{each edge in $e \in r_0$}{ choose an arbitrary ordering of the vertices in $e$ and denote each ordered pair as $a_e$.}
Choose a random ordering of $\{a_e \ | e \in r\}$ which induces the sequence ${\bf a}$.\;
Choose a random composition $\sigma$ of $\# r$ such that the size of every part of $\sigma$ is strictly greater than one.  The composition $\sigma$ should be chosen according to a known but arbitrary distribution $P_{\# r} (\sigma)$. Let $k$ be the length of $\sigma$ and partition ${\bf a}$ into $k$ subsequences according to the composition $\sigma$, ${\bf a}=({\bf a}_1, {\bf a}_2, \ldots, {\bf a}_k)$.
\For{$1 \leq j \leq k$}{ let ${\bf a}_j=(a_{e_1}, \ldots, a_{e_m})$. Let $b_j$ be the set of edges obtained by joining the head of $a _{e_{i+1}}$ with the tail of $a_{e_{i}}$ for $i$ from 1 to $m-1$ and joining the head of  $a_{e_1}$ with the tail of $a_{e_m}$,
$$b_j:=\{ \ (a_{e_{i+1}}(2), a_{e_{i}}(1) ) \ | 1 \leq i < m \ \} \cup \{ \ (a_{e_{1}}(1), a_{e_m}(2))\ \}.$$}
Let $b= \cup_{j=1}^k \text{ recip}(b_j)$\;
Let $r=\text{ recip}(r_0)$\;
\uIf
{($b$, as a graph, is not simple) {\bf or} ($b \cap (E(g_u)-r) \neq \emptyset$) {\bf or} ($b \cap \text{undir}(g_d) \neq \emptyset$)}
{return the trivial move $(\emptyset, \emptyset)$}
\uElse{ return $(b, r)$.}
\caption{Generating a Type 1 Move} 
\end{algorithm}
}

\smallskip

{\footnotesize  
\begin{algorithm}[H]
\label{alg:type2move}
\LinesNumbered
\DontPrintSemicolon
\SetAlgoLined
\SetKwInOut{Input}{input}
\SetKwInOut{Output}{output}
\Input{
$g_d$, a directed graph}
\Output{$(b,r)$, a type 2 (non-reciprocated-only) applicable move.}
\BlankLine
Choose a random subset $r$ of edges from $g_d$.  \;
Choose a random ordering of $\{a_e \ | e \in r\}$ which induces the sequence ${\bf a}$.\;
 Choose a random composition $\sigma$ of $\# r$ such that the size of every part of $\sigma$ is strictly greater than one.  The composition $\sigma$ should be chosen according to a known but arbitrary distribution $P_{\# r} (\sigma)$. Let $k$ be the length of $\sigma$ and partition ${\bf a}$ into $k$ subsequences according to the composition $\sigma$, ${\bf a}=({\bf a}_1, {\bf a}_2, \ldots, {\bf a}_k)$.\;
\For{$1 \leq j \leq k$}{ ${\bf a}_j=(a_{e_1}, \ldots, a_{e_m})$. Let $b_j$ be the set of edges obtained by joining the head of $a _{e_{i+1}}$ with the tail of $a_{e_{i}}$ for $i$ from 1 to $m-1$ and joining the tail of  $a_{e_1}$ with the head of $a_{e_m}$,
$$b_j:=\{ \ (a_{e_i}(2), a_{e_{i+1}}(1) ) \ | 1 \leq i < m \ \} \cup \{ \ (a_{e_{m}}(2), a_{e_1}(1))\ \}.$$}
 Let $b= \cup_{j=1}^k b_j$\;
\uIf{undir$(b)$, as a graph, is not simple, {\bf or} $b \cap (E(g_d)-r) \neq \emptyset$ {\bf or} $b \cap \text{undir}(g_u) \neq \emptyset$ }{return the trivial move $(\emptyset, \emptyset)$.}
\uElse{return $(b, r)$.}
\caption{Generating a Type 2 Move.} 
\end{algorithm}
}

\smallskip

{\footnotesize  
\begin{algorithm}[H]
\label{alg:type3move}
\LinesNumbered
\DontPrintSemicolon
\SetAlgoLined
\SetKwInOut{Input}{input}
\SetKwInOut{Output}{output}
\Input{
$g=g_u \cup g_d, $, a directed graph}
\Output{$(b,r)$, a type 3 (mixed) applicable move.}
\BlankLine
Perform steps 1-9 of Algorithm~\ref{alg:type1move} to obtain $(b_u, r_u)$.\;
Perform steps 1-7 of Algorithm~\ref{alg:type2move} to obtain $(b_d, r_d)$.\;
\uIf{
$b_u$, as a graph, is not simple, {\bf or}
undir$(b_d)$, as a graph, is not simple, {\bf or}
$b_u \cap $ undir$b_d \neq \emptyset$, {\bf or}
$b_u \cap (E(g_u)-r) \neq \emptyset$, {\bf or}
$b_d \cap (E(g_d)-r) \neq \emptyset$, {\bf or}
$b_u \cap \text{undir}(g_d) \neq \emptyset$, {\bf or}
$b_d \cap \text{undir}(g_u) \neq \emptyset$
}
{return the trivial move $(\emptyset, \emptyset)$, }
\uElse{return $(b_u \cup b_d, r_u \cup r_d)$.}
\caption{Generating a Type 3 Move} 
\end{algorithm}
}

\begin{example}
Figure~\ref{fig:moveExample}  illustrates the process of generating a Type 2 move.  First the edges $(x_2, x_1)$, $(x_3, x_4)$, and $(x_5, x_6)$ from a network $g$ are chosen.  These will be the edges that are removed from $g$ in the move. We consider these edges as edges of $A_n$.  A walk is completed on $A_n$ by adding the blue edges $\{\alpha_2, \beta_6\}, \{\alpha_3, \beta_1\}$, and $\{ \alpha_5, \beta_4\}$. The blue edges are then interpreted in terms of pairs and dyadic configurations in $g$.  These are the edges that are added to $g$ in the move.
\begin{figure}
\begin{center}
\includegraphics[width=5in]{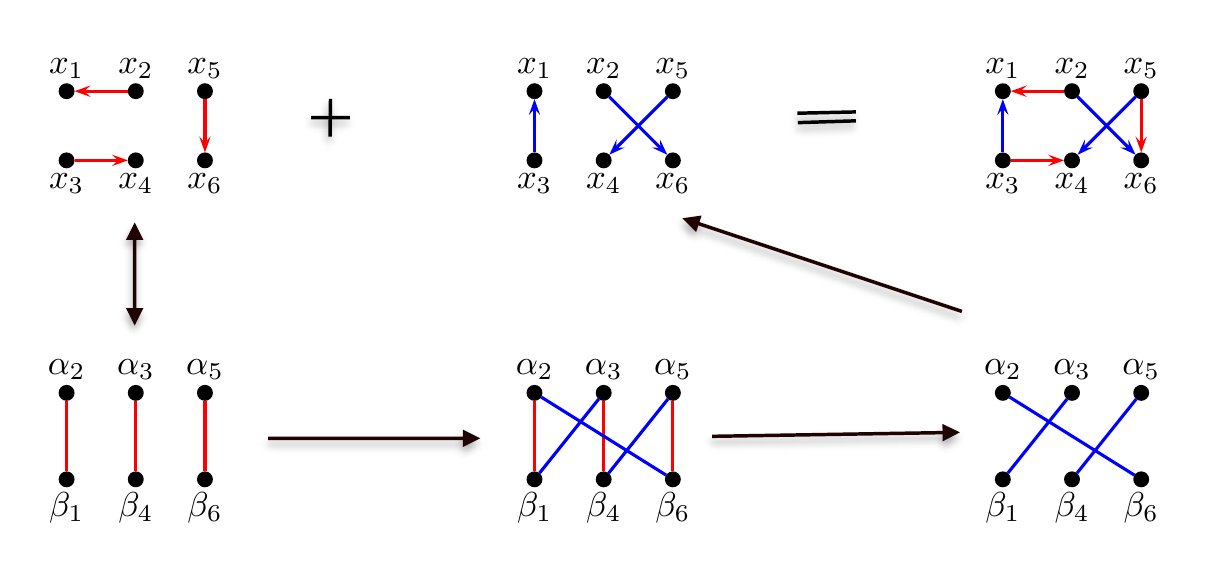}
\end{center}
\caption{An example of generating a Type 2 move. }
\label{fig:moveExample}
\end{figure}
\end{example}

\begin{remark} Notice that in each of the above algorithms, it is possible that the trivial move is returned.  This means the walk in Algorithm \ref{alg:MH} would stay in the same place at that step.  While this does not affect the stationary distribution of the Markov chain, it can have a negative impact on mixing times if too many trivial moves are returned.  However, this is the problem also with the usual Metropolis-Hastings algorithm, as mixing time questions are generally open. Section~\ref{sec:examples} shows some indication that the chain seems to be mixing well.  In the case of the $p_1$ model, the probability of returning the trivial move in any of the above algorithms depends on the in and out-degree sequences of the unreciprocated edges and the reciprocated edges.  One direction for further research is to understand and try and reduce the output of trivial moves.  Even understanding which networks result in a high probability of a trivial move being returned in Algorithms \ref{alg:type1move}, \ref{alg:type2move}, \ref{alg:type3move}, would be an interesting combinatorial problem.  
\end{remark}

\begin{prop}\label{prop:moves}  Every move outputted by Algorithms \ref{alg:type1move}, \ref{alg:type2move}, \ref{alg:type3move} is an applicable Markov move of the form \emph{lift}$W$ such that $W$ is a balanced edge set on $A_n \cup K_n$ and   \mbox{$\deg_{\mathcal R}(\lambda_{ij})=\deg_{\mathcal B}(\lambda_{ij})\in\{0,1\}$.} Moreover, on input $g_1$, if $g_2 \in \fiberofg$ and $g_1 \neq g_2$, Algorithm \ref{alg:wrapper} has  a non-zero probability of returning the move $g_2-g_1$.
\end{prop}

\begin{proof} Algorithm \ref{alg:type1move} chooses a set of edges $r_0$ from undir$(g_u)$ and completes $k$ closed even walks on $K_n$.  We will denote the balance edge set of $A_n$ corresponding to this set of closed even walks as $W$. Step 7 checks that lift$W=(\cB, \cR)$  satisfies \mbox{$\deg_{\mathcal R}(\lambda_{ij})=\deg_{\mathcal B}(\lambda_{ij})\in\{0,1\}$}.  If the condition is not satisfied, then the trivial move is returned. Otherwise, $(r,b)$ outputted by Algorithm \ref{alg:type1move} is of the form of the specified.  Applicability of $(b, r)$ follows from the fact that $r$ is a subset of $g_u$ and $\deg_{\mathcal R}(\lambda_{ij})=\deg_{\mathcal B}(\lambda_{ij}) \leq 1$.  Moves outputted from Algorithms \ref{alg:type2move}, \ref{alg:type3move} can be analyzed in a parallel fashion.

For the second part of the statement, Proposition~\ref{prop:lift} states that the move between two networks $g_1$, $g_2$ in the same fiber is of the form lift$W$ where $W=(R,B)$ is a balanced edge set on $A_n \cup K_n$. Assume that $R$ is contained entirely in $K_n$.  Denote the closed even walks on $K_n$ that correspond to $W$ as $W_1, \ldots, W_k$. The move $g_2-g_1$ will be returned if $1$ is chosen in Algorithm \ref{alg:wrapper}, the edges of $g_1$ corresponding to $R$ are chosen at Step 1 of Algorithm \ref{alg:type1move}, and Steps 3 and 4 result in a sequence ${\bf a}=({\bf a}_1, {\bf a}_2, \ldots, {\bf a}_k)$ such that ${\bf a}_i$ corresponds to a cyclic permutation of the odd edges of $W_i$. If $R$ is contained entirely in $A_n$ or contains edges from both $A_n$ and $K_n$, then a similar argument follows.
\end{proof}

\begin{theorem}  Let $g$ be an observable network with more than 2 edges and with sufficient statistic $s$. The Markov chain, $(\mathcal{G}_t)_{t=0}^{\infty}$, where the step from $\mathcal G_i$ to $\mathcal{G}_{i+1}$ is given by Algorithm~\ref{alg:wrapper} is an irreducible, symmetric, and aperiodic random walk on $\fiber$.
\end{theorem}

\begin{proof}
Irreducibility follows Proposition \ref{prop:moves}. 

To show symmetry, let $g_1=g_{1_u} \cup g_{2_d}$ and $g_2=g_{2_u} \cup g_{2_d}$ be two simple networks with reciprocated parts $g_{1_u}$, $g_{2_u}$ and unreciprocated parts $g_{1_d}$, $g_{2_d}$. The move $(b, r)$ from $g_1$ to $g_2$ is the combination of moves $(b_u, r_u)$ from $g_{1_u}$ to $g_{2_u}$ and $(b_d, r_d)$ from $g_{1_d}$ to $g_{2_d}$ where $b=b_u \cup b_d$ and $r=r_u \cup r_d$.  The move $(b_u, r_u)$ corresponds to a balanced edge set $W_u=(B_u, R_u)$ on $K_n$, which forms a set of primitive closed even walks on $K_n$.  The probability of choosing $r_u$ in Step 1 of Algorithm \ref{alg:type1move} is dependent only on the number of edges in $g_{1_u}$, which is equal to the number of edges in $g_{2_u}$. Step 5 in Algorithm \ref{alg:type1move} completes walks on sequences of edges from $r_u$ by connecting heads to tails. Thus, given that $r_u$ was chosen in Step 1, the probability of choosing an ordering of the vertices, an ordering of the edges,  and a composition in Steps 2-4 such that Step 5 will output $b_u$ is dependent only on the structure of $W_u$ (the primitive walks in $W_u$, the length of these walks, and which of these walks share a vertex). So, since $W_u$ is the same regardless whether we are moving from $g_{1_u}$ to $g_{2_u}$ or from $g_{2_u}$ to $g_{1_u}$, the probabilities of making these moves in a single step are equal. A similar situation occurs between the reciprocated parts of $g_1$ and $g_2$.

For aperiodicity, notice that every non-diagonal entry of the transition matrix $P$ of $(\mathcal{G}_t)_{t=0}^{\infty}$ is greater than zero.  Therefore, since $g$ contains more than two edges, $P^n(i,j) > 0$ for all $n\geq 2$.
\end{proof}

\begin{corollary} If $g$ has more than two edges, then with probability one
\[ \lim _{N \to \infty} \frac{1}{N} 1_{\chi^2(\mathcal{G}) \geq \chi^2(g)} = P( \chi^2(\mathcal{G}) \geq \chi^2(g) \ : \mathcal G \in \fiberofg ).\]
\end{corollary}

Algorithm~\ref{alg:wrapper} and it's subroutines Algorithms~\ref{alg:type1move}, \ref{alg:type2move}, \ref{alg:type3move} are implemented in {\tt R}; the code is available  in the supplementary material on \cite{supplement}. 
The examples in Section~\ref{sec:examples} that compute estimated $p$-values use the function {\tt Estimate.p.Value}.   It takes an observed network and implements Algorithm~\ref{alg:MH} using an iterative proportional scaling algorithm \cite[p.40]{HL81} to compute the MLE,  and Algorithm~\ref{alg:wrapper} for Step 4. We chose to use the chi-square statistic for the goodness-of-fit statistic. 

Our implementation makes use of the R package igraph \cite{igraph}, and in particular its graph data structure and methods for producing graph unions and graph intersections. Each of these methods has complexity linear in the sum of the cardinalities of the edge sets and vertex sets of the input. As a result the complexity of the algorithm is at worst   O$\left(\left(|V|+|R|\right)^2\right)$, where V and E are the vertex and edge sets respectively.

\section{Simulations} 
\label{sec:examples}

We apply Algorithms~\ref{alg:MH} and~\ref{alg:wrapper} and run goodness-of-fit tests in {\tt R}  on several real-world network datasets as well as simulated networks under   the $p_1$ model. 
 In what follows,  reported are the number of steps in the chain along with the initial burn-in. 
Our statistic of choice for $GF(u)$ is the chi-square statistic, directly measuring the distance of the network $u$ from the 
 MLE.  For each simulation, we report  the estimated $p$-value returned on line 11 of Algorithm~\ref{alg:wrapper}  and the sampling distribution of $GF(u)$.

\subsection{A  small  synthetic network} 
\label{sec:OHT}

We begin with a test case to check how  Algorithm \ref{alg:wrapper}  explores the fiber.  In \cite[\S 5.1]{OHT}, the authors sample the fiber of an undirected graph $H_0$ on $8$ nodes, depicted in Figure~\ref{fig:OHTgraphPlot}, under the beta model.  By enumeration they have determined that the size of the fiber is $591$.  Considering this graph as a directed network all of whose edges are reciprocated, we can test the fit of the $p_1$ model as well, and study its fiber similarly. The fibers of  $H_0$  under the two models are the same, since in both cases, the fiber consists of all undirected (or reciprocated-edge) graphs with the same (in- and out-) degree vector as $H_0$.

\begin{figure}[ht]
\centering
\begin{minipage}[b]{0.3\linewidth}
	\centering 
	\includegraphics[width=\linewidth]{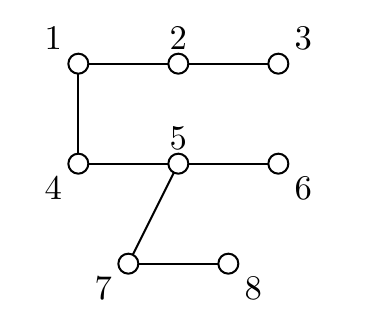}
	\caption{The graph $H_0$ from Figure 13 in \cite{OHT}}
	\label{fig:OHTgraphPlot}
\end{minipage}
\quad
\begin{minipage}[b]{0.65\linewidth}
	\centering
	\includegraphics[width=.7\linewidth]{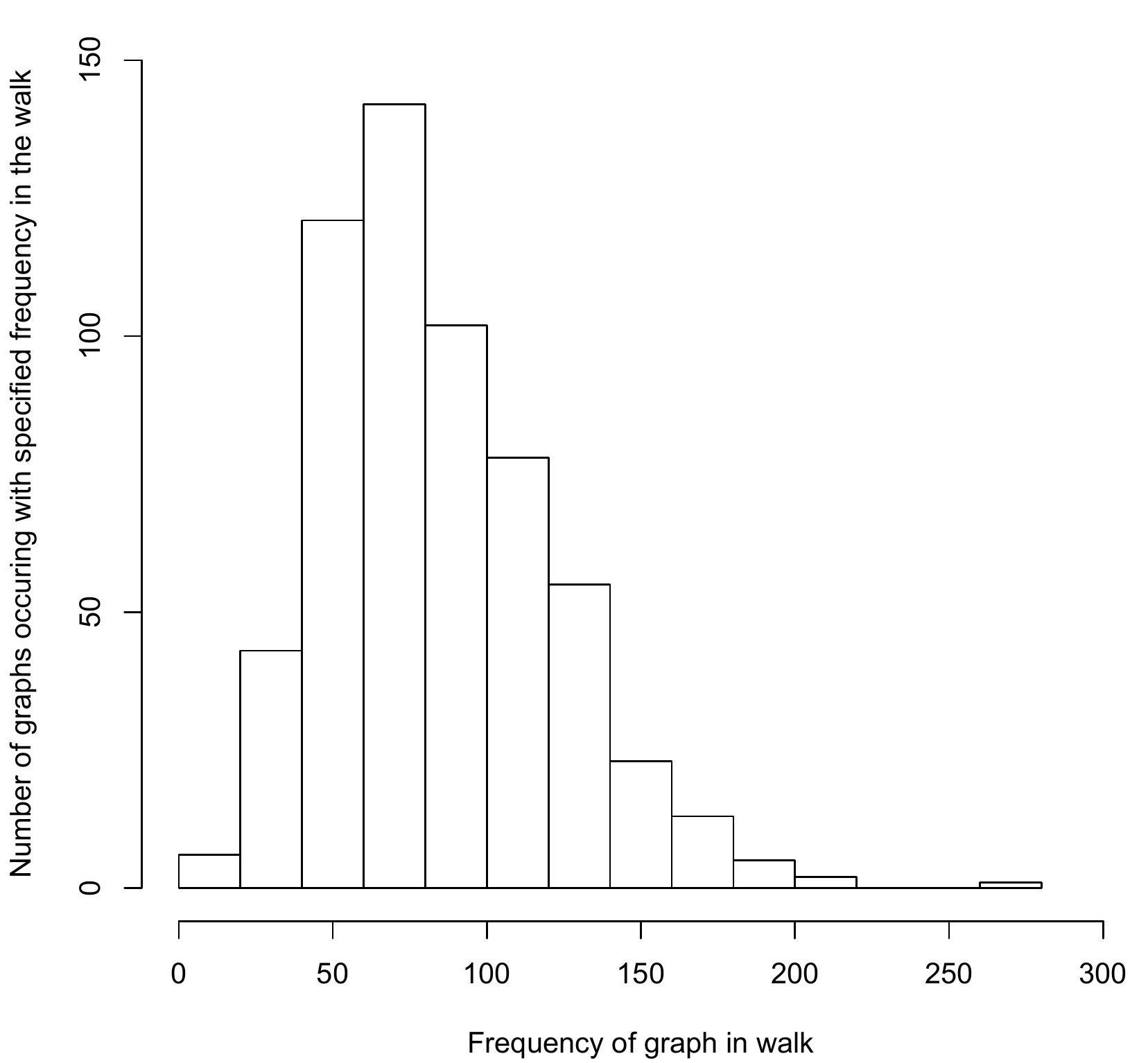}
	\caption{Histogram from sampling}  
	\label{fig:OHTsamplingHistogram}
\end{minipage}
\end{figure}

We ran Algorithm~\ref{alg:wrapper} and stored all graphs discovered in the run. Starting from $H_0$, after $1,000$ steps, $232$ points in the fiber were discovered. 
After $5,000$ steps, $538$ graphs were discovered; and the entire fiber of $591$ graphs  was reached after less than  $15,000$ steps in the chain. 
At this point, the chain  samples the fiber almost uniformly, as the total variation distance between the sampling distribution and the uniform distribution on the fiber is calculated to be $0.2088025$ (at the $15,000$-th step).   For comparison purposes, the TV-distance is  $0.1703418$  after $50,000$ steps; Figure~\ref{fig:OHTsamplingHistogram} shows the histogram of graphs sampled in the $50,000$-move walk. 
 Therefore, running a Markov chain of at least $50,000$ steps should be sufficient for testing purposes for this example.

A run of Algorithm~\ref{alg:MH} for $450,000$ steps, after $50,000$ burn-in steps, produced the values of the chi-square statistics in Figure~\ref{fig:OHThistogram-gof-500k}, and the $p$-value estimate of $0.86$. The estimates of the p-value from the simulation are  plotted in Figure~\ref{fig:OHTpvalues-500k} against the step number of the Markov chain and give further evidence of convergence.
\begin{figure}
\centering
\begin{subfigure}{.45\textwidth}
	\includegraphics[width=1\linewidth]{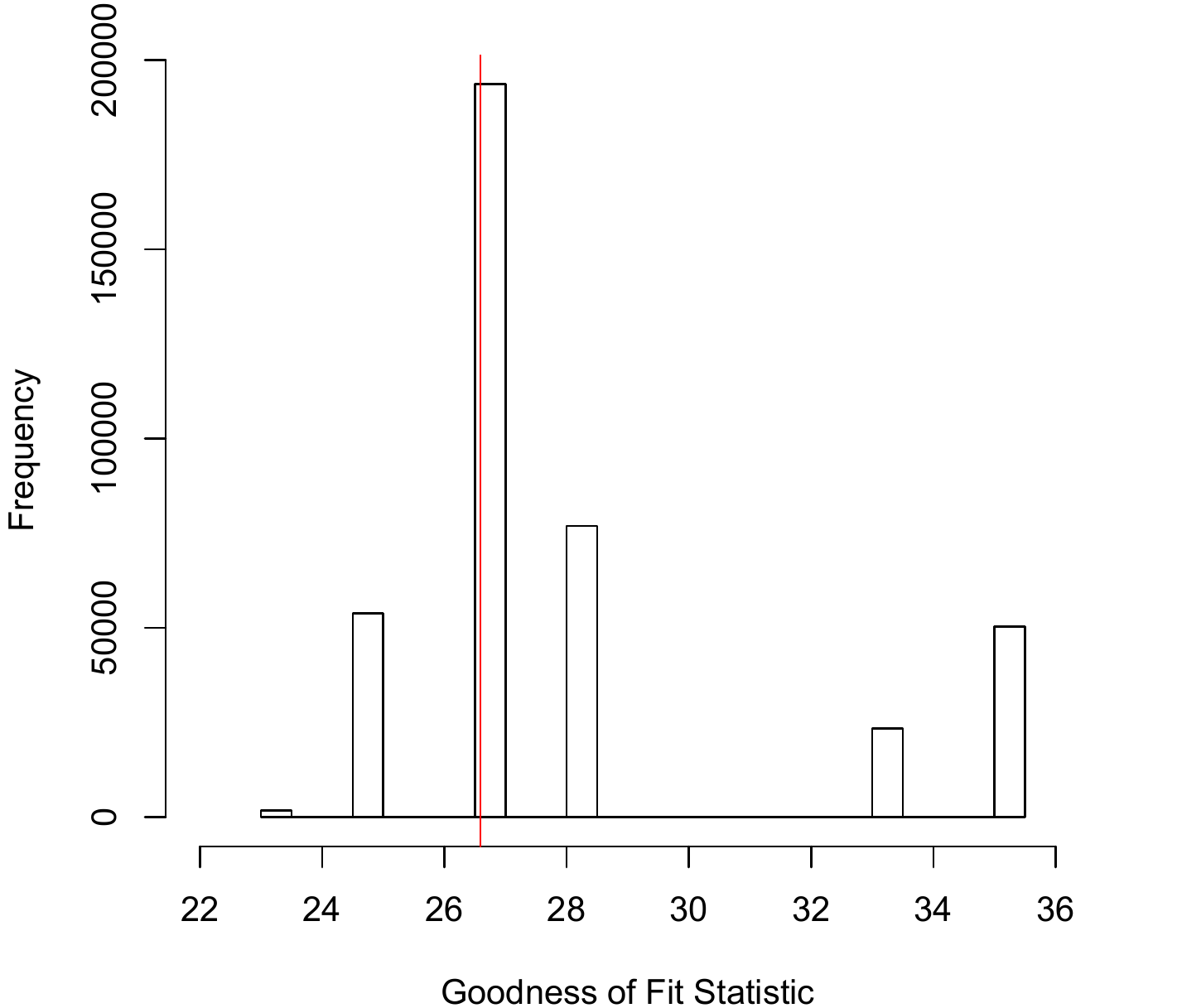}
	\caption{Histogram of chi-square statistic}
	\label{fig:OHThistogram-gof-500k}
\end{subfigure}
\quad
\begin{subfigure}{.45\textwidth}
\centering
	\includegraphics[width=1\linewidth]{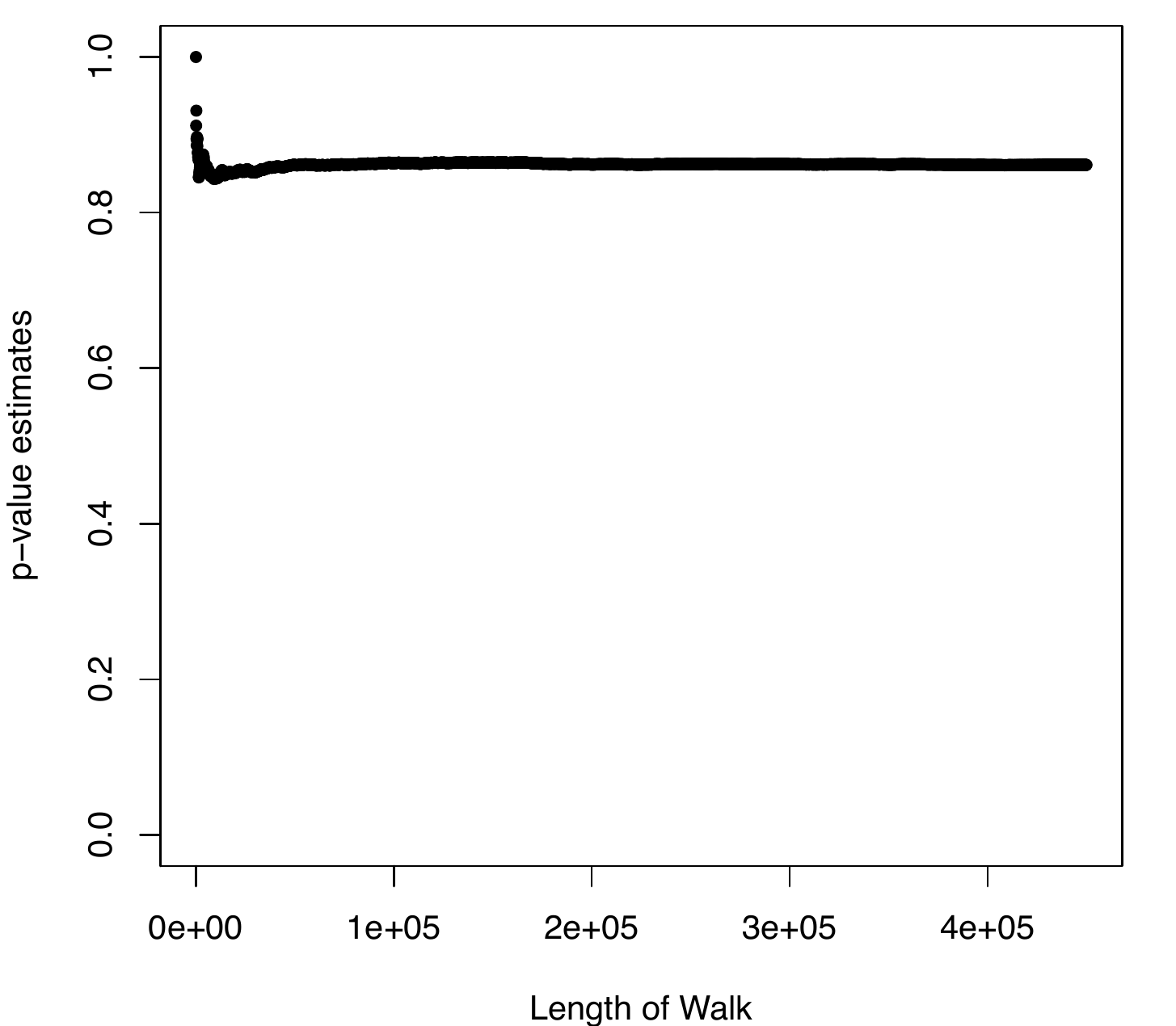}
	\caption{The $p$-value estimates.}
	\label{fig:OHTpvalues-500k}
\end{subfigure}
\caption{Simulation results for graph $H_0$: chain of length $500,000$ including $50,000$ burn-in steps.}
\end{figure}

\subsection{Networks simulated from the $p_1$ distribution} 
\label{sec:hl}
Consider  the four digraphs on $10$ nodes that Holland and Leinhardt simulated  from the $p_1$ distribution; see~\cite[Figure 3]{HL81}. 
The networks  are depicted in Figure~\ref{fig:HLgraphs}. 
\begin{figure}
\centering
\begin{subfigure}{0.24\linewidth}
  \centering
  \includegraphics[width=.8\linewidth]{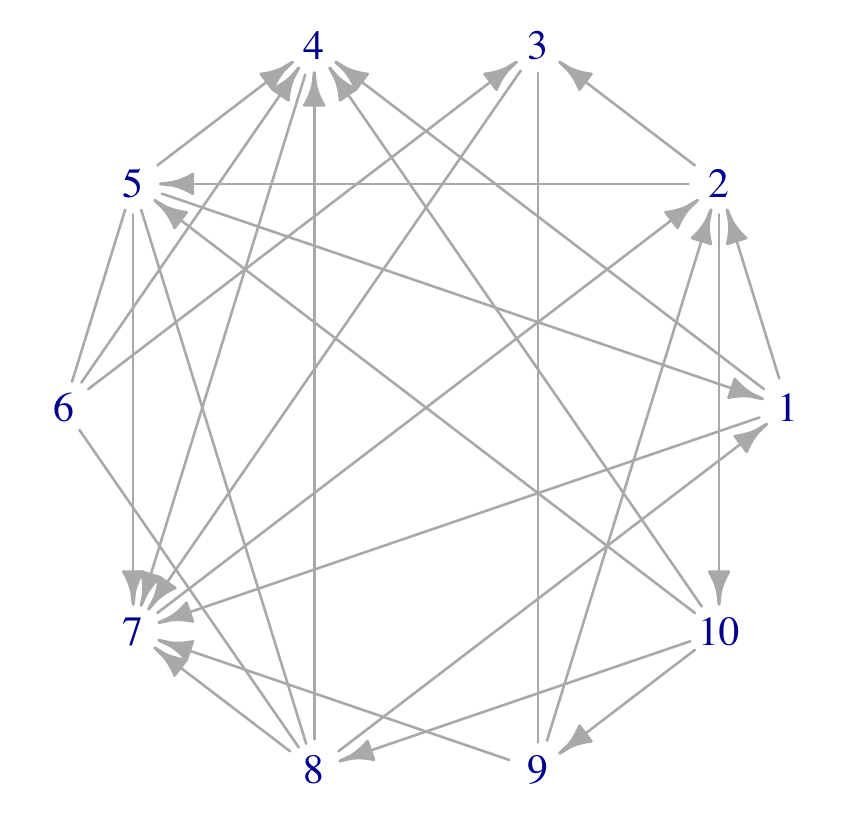}
  \caption{Graph 1}
	\label{fig:HL1}
\end{subfigure}
\begin{subfigure}{0.24\linewidth}
  \centering
  \includegraphics[width=.8\linewidth]{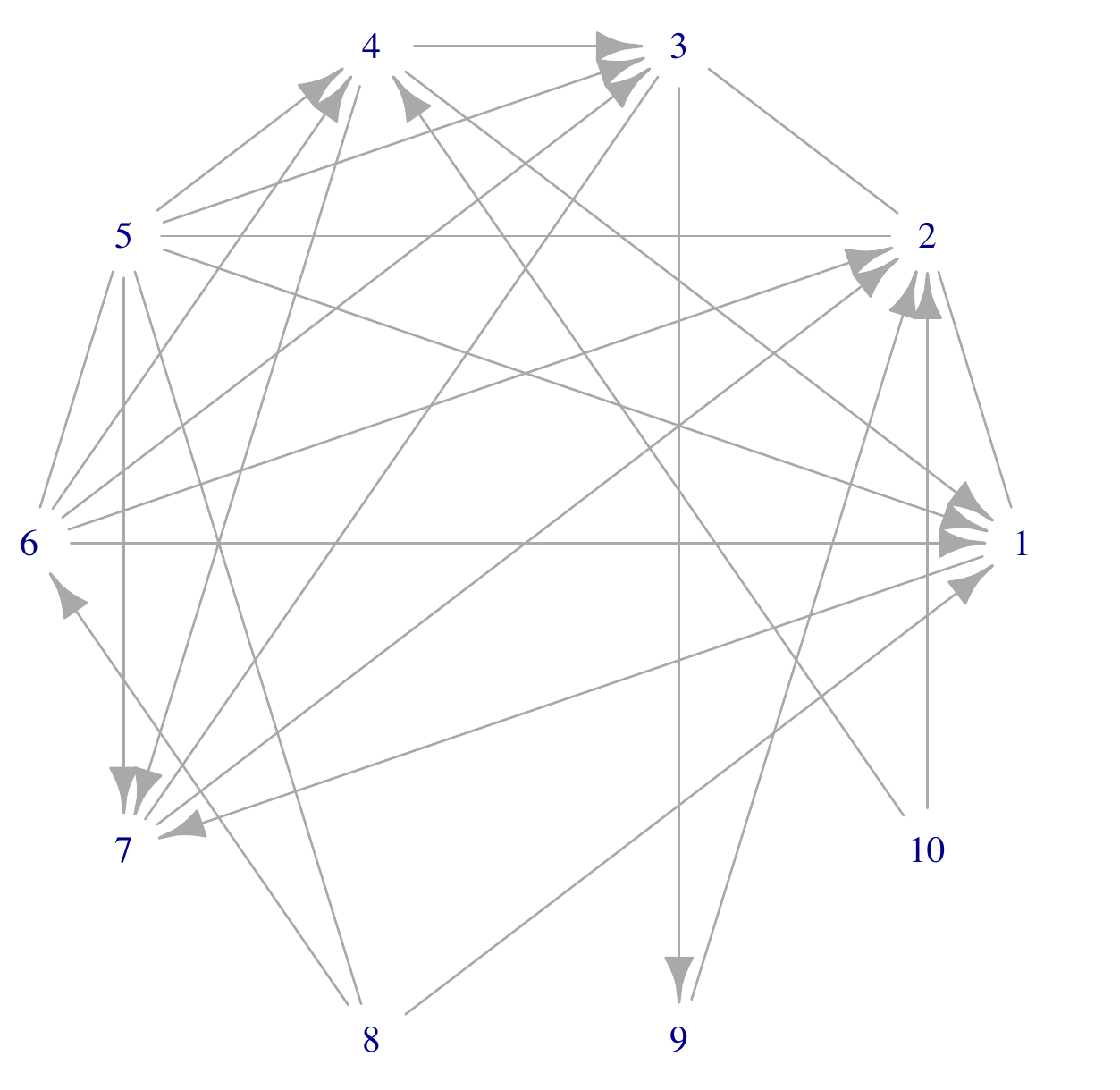}
  \caption{Graph 2}
	\label{fig:HL2}
\end{subfigure}
\begin{subfigure}{0.24\linewidth}
  \centering
  \includegraphics[width=.8\linewidth]{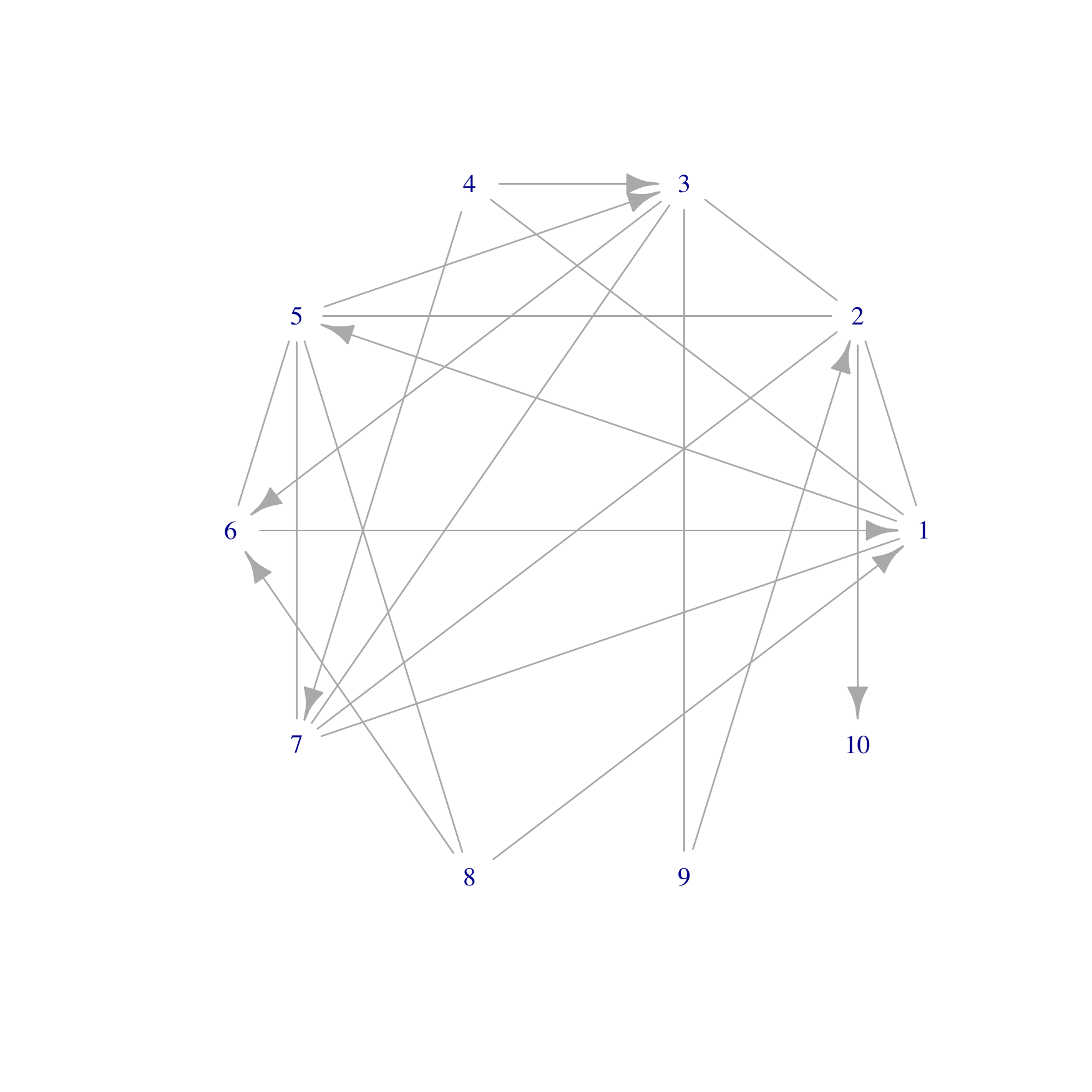}
  \caption{Graph 3}
	\label{fig:HL3}
\end{subfigure}
\begin{subfigure}{0.24\linewidth}
  \centering
  \includegraphics[width=.8\linewidth]{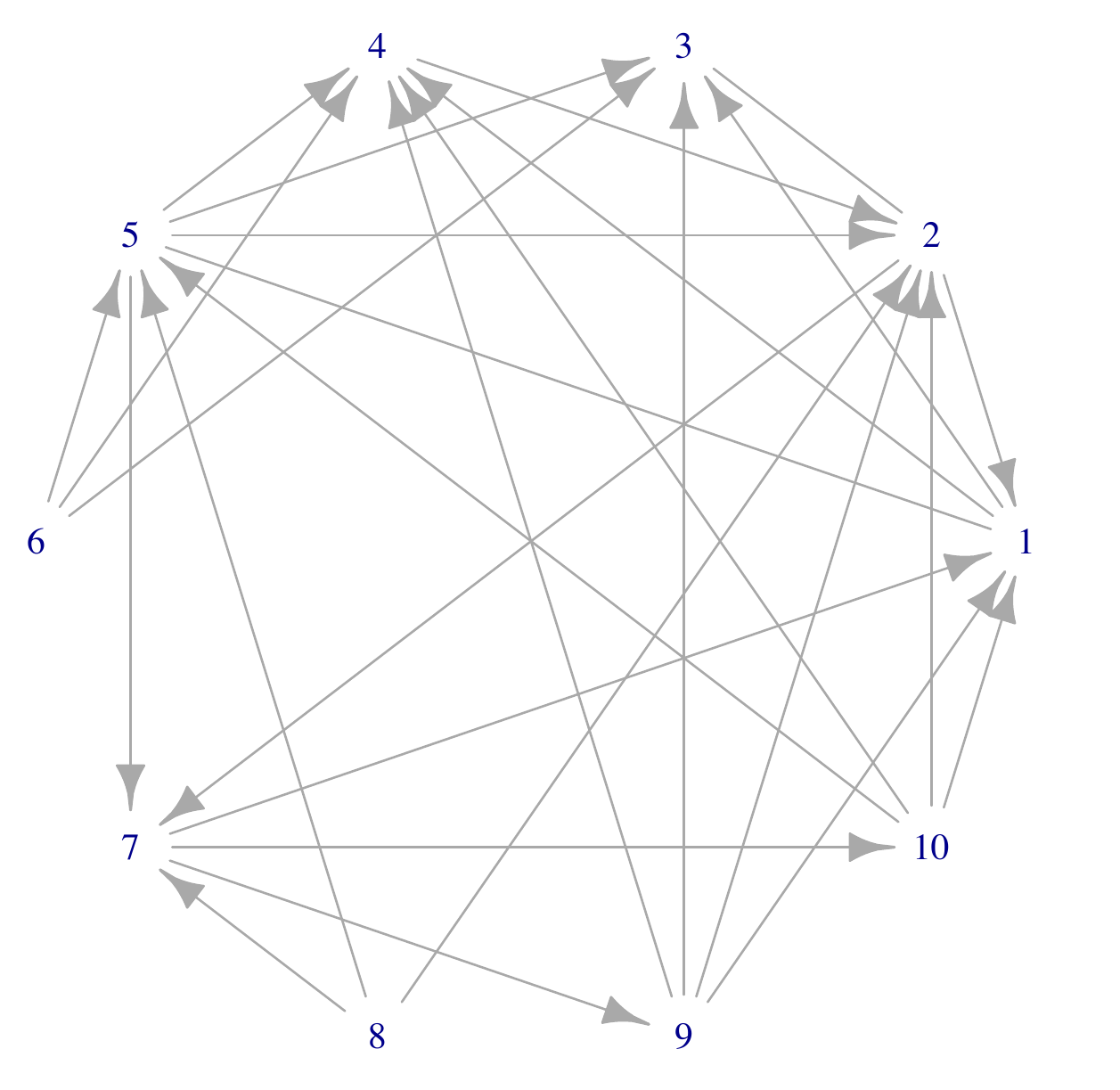}
  \caption{Graph 4}
	\label{fig:HL4}
\end{subfigure}
\caption{Four digraphs simulated from the $p_1$ distribution from \cite[Figure 3]{HL81}. For clarity, the reciprocated edges are drawn as undirected.} 
\label{fig:HLgraphs}
\end{figure}

For each network, chains of length $200,000$ provide expected results. The estimated p-values are $0.284774$, $0.7185896$, $0.4673885$ and $0.7432897$, respectively. 
The histograms of the sampling distribution of the  chi-square  statistics from the $220,000$-step simulation (with $20,000$ burn-in steps) are shown in Figure~\ref{fig:HLhistograms}.  
The p-values reach their estimated value in approximately $25,000$ steps after burn in. 

\begin{figure}
\centering
\begin{subfigure}{.25\textwidth}
  \centering
	\includegraphics[width=1\linewidth]{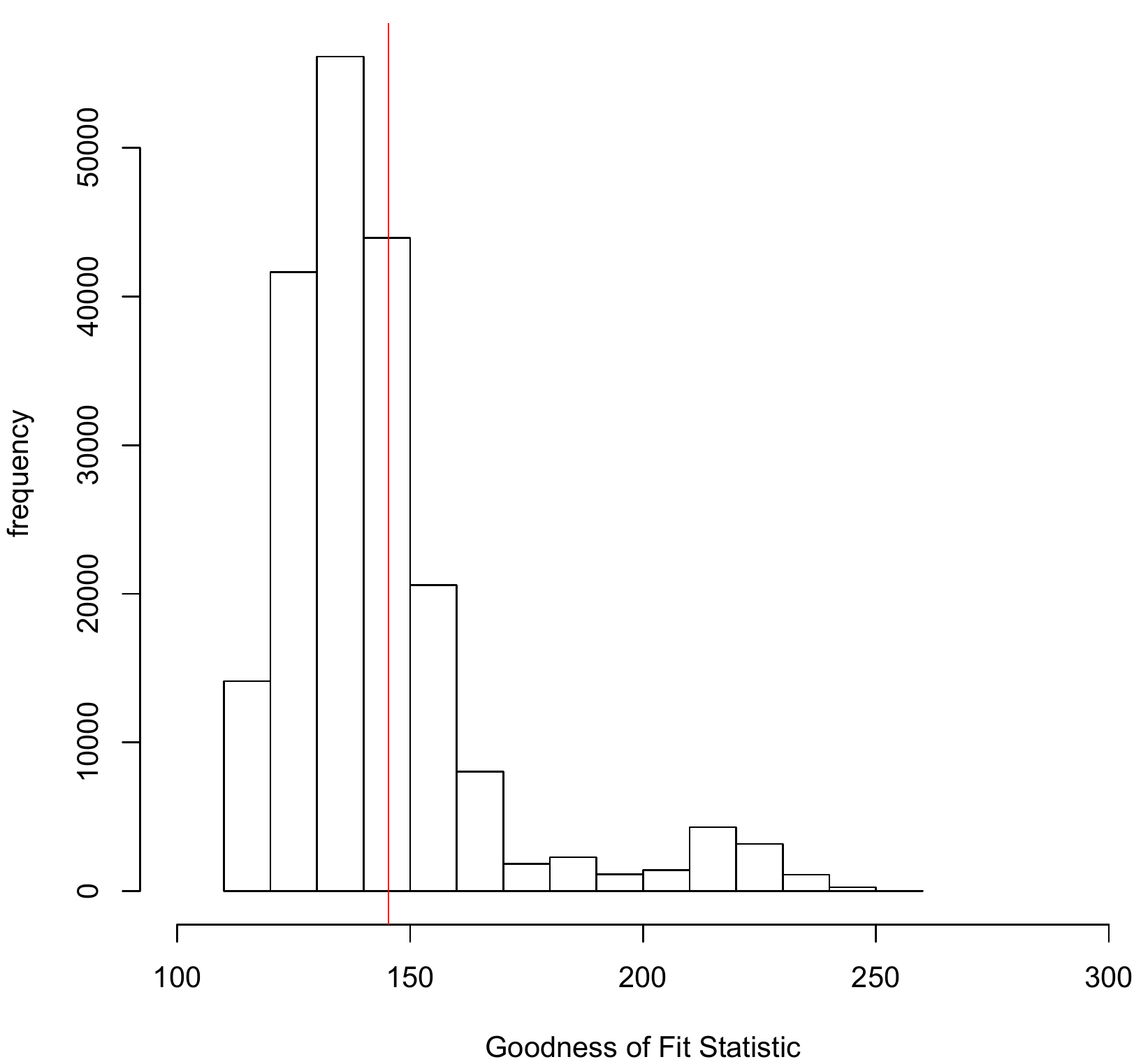}
  \caption{ Graph 1.\\ $p$-value: $0.284774$}
	\label{fig:HL1histogram}
\end{subfigure}%
\begin{subfigure}{.25\textwidth}
  \centering
  \includegraphics[width=1\linewidth]{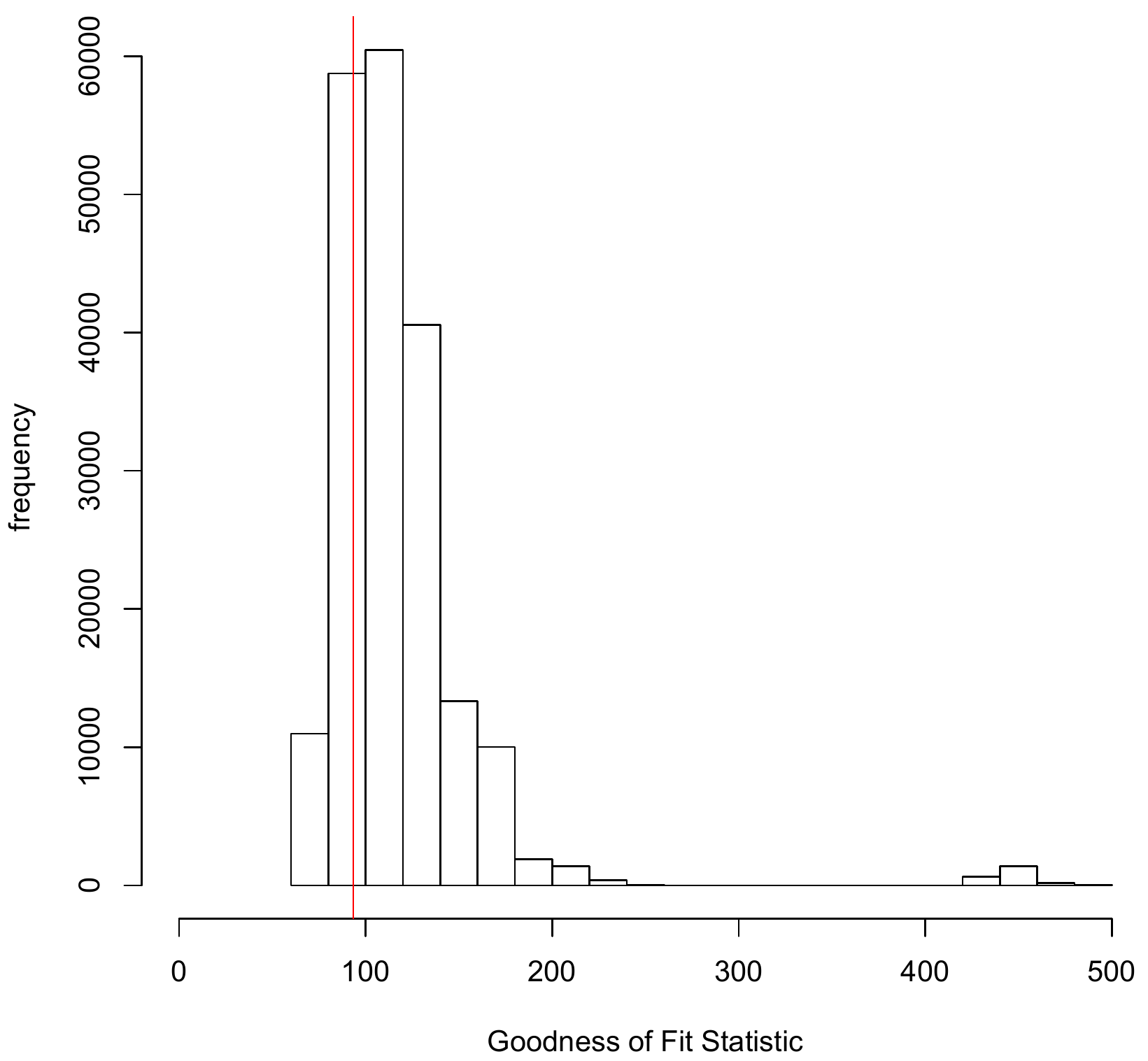}
  \caption{Graph 2. \\$p$-value: $0.7185896$}
	\label{fig:HL2histogram}
\end{subfigure}
\begin{subfigure}{.24\textwidth}
  \centering
  \includegraphics[width=1\linewidth]{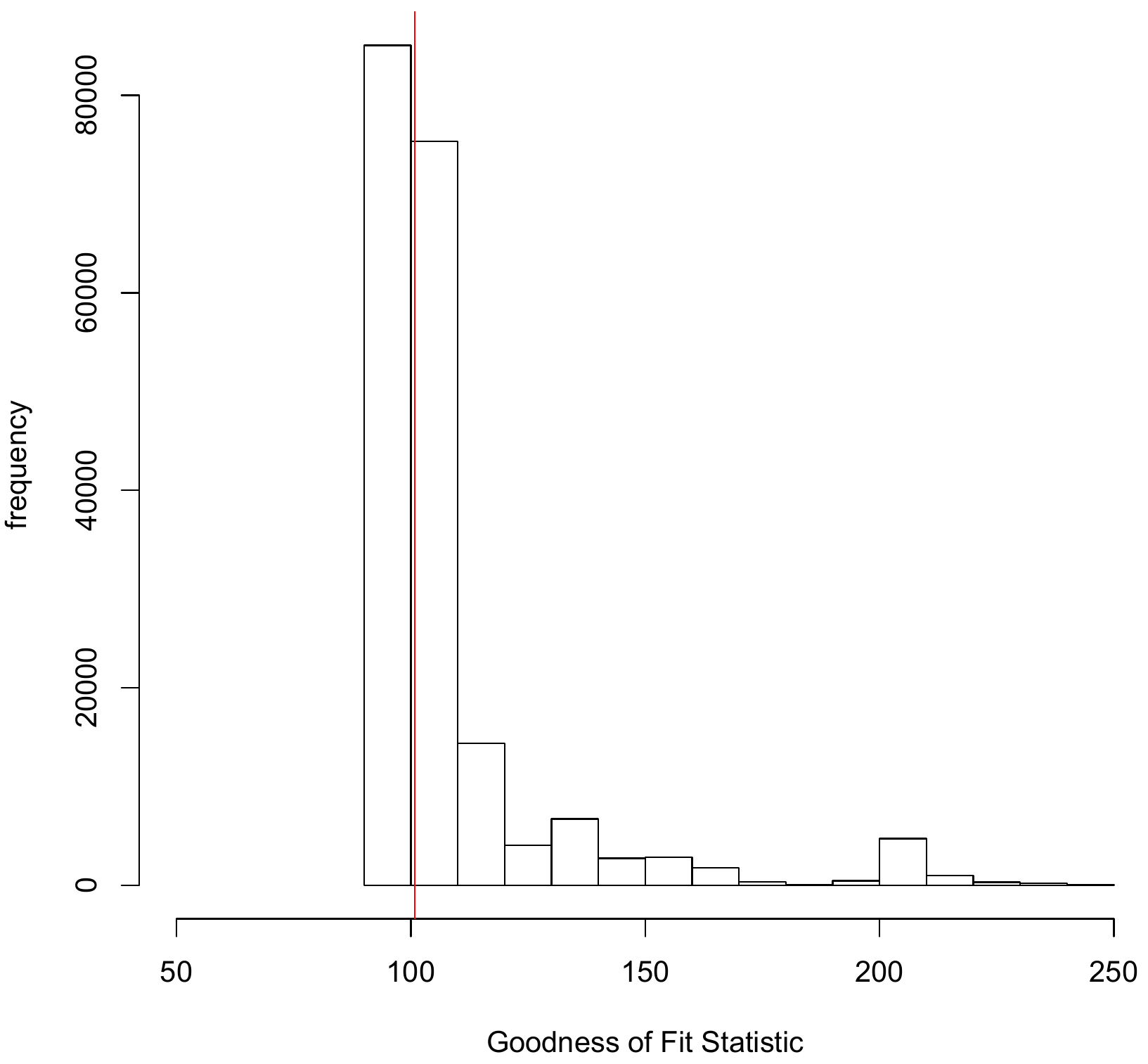}
  \caption{Graph 3. \\$p$-value: $0.4673885$}
	\label{fig:HL3histogram}
\end{subfigure}
\begin{subfigure}{.24\textwidth}
  \centering
  \includegraphics[width=1\linewidth]{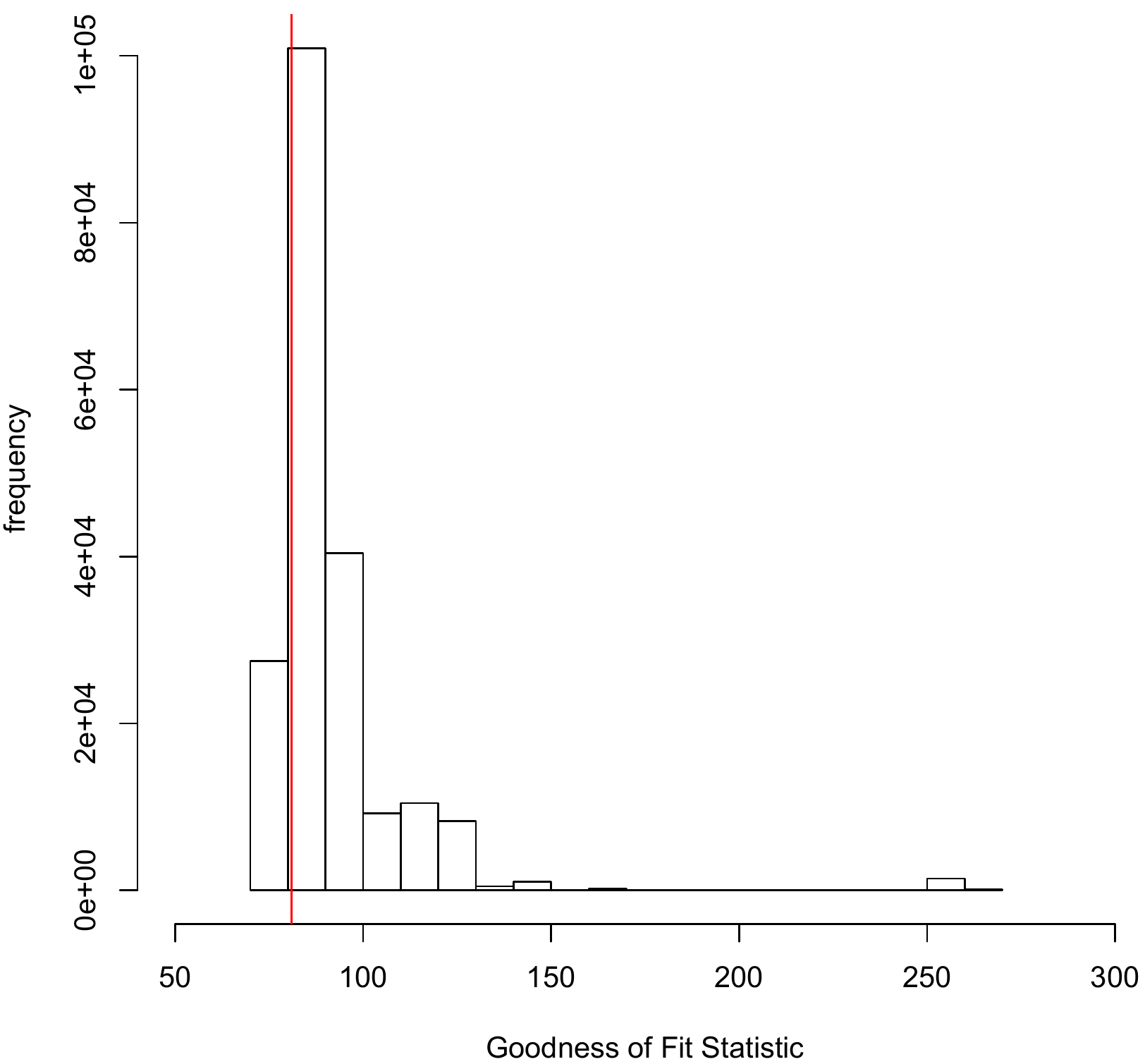}
    \caption{Graph 4.\\ $p$-value: $0.7432897$}
	\label{fig:HL4histogram}
\end{subfigure}
\caption{Histograms of  chi-square  statistics from sampling the fibers with $220,000$ steps ($20,000$ burn-in steps) for the four digraphs in Figure~\ref{fig:HLgraphs}.}
\label{fig:HLhistograms}
\end{figure}

\subsection{Mobile money networks}
\label{sec:kenya}

Figure \ref{fig:kenyaNetwork}  is a directed graph on 12 vertices with 13 unreciprocated edges and 15 reciprocated edges.  The data is from \cite{KCGK13} and was collected through a survey conducted in Bungoma and Trans-Nzoia Counties in Kenya, and among Kenyans living in Chicago, Illinois in the summer of 2012. Vertices represent members of an extended family. An edge from vertex $v_i$ to vertex $v_j$ represents that $v_i$ had sent money to $v_j$ using a mobile money transfer.  Since the network depicted in Figure \ref{fig:kenyaNetwork} is a social network and the individuals are social actors, it is reasonable to suspect transitive effects are present.  In such a setting, it is expected the $p_1$ model would not fit this data very well, and, Holland and Leinhardt suggest  \cite{HL81} the $p_1$ model as a realistic null model in such cases.  
\begin{figure}
\begin{minipage}[b]{0.3\linewidth}
	  \includegraphics[width=1\linewidth]{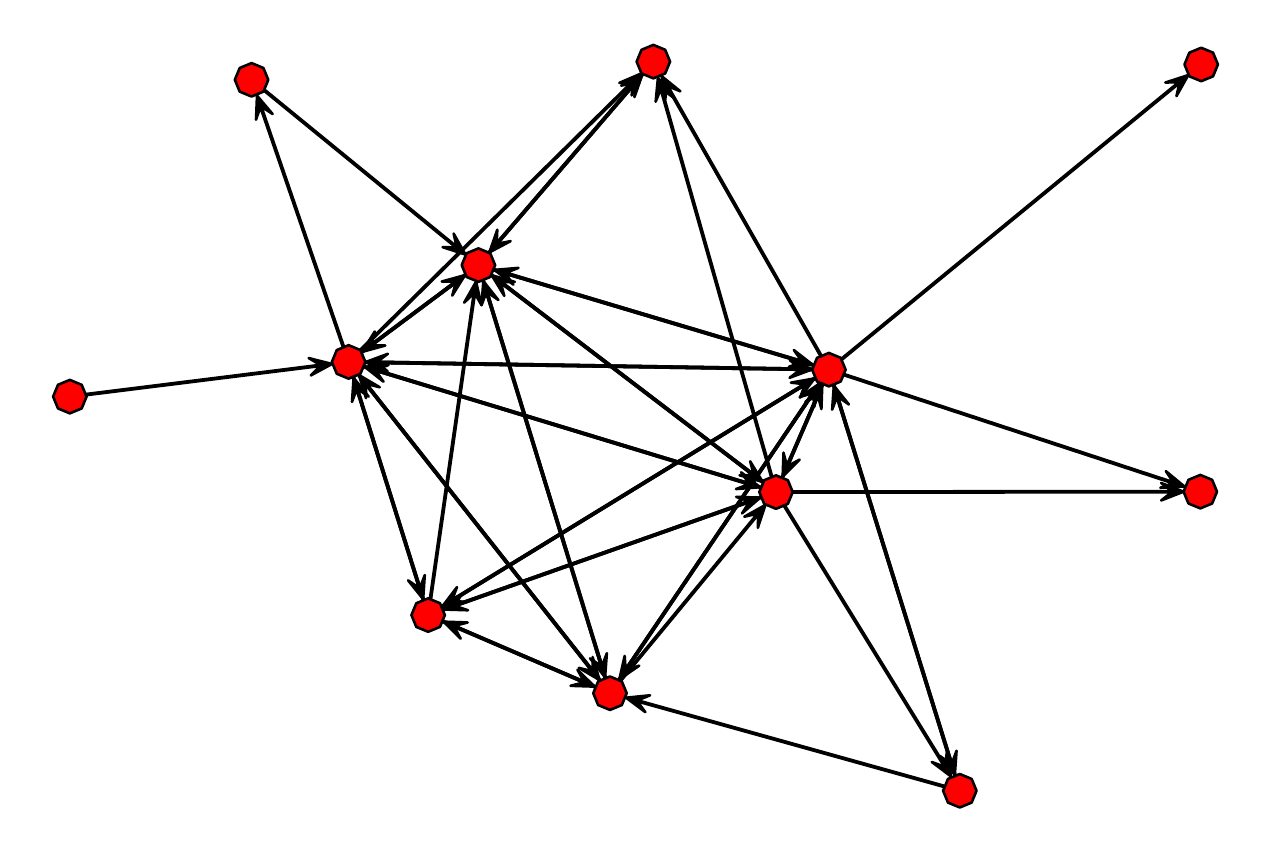}
	  \caption{Mobile money transfers between members of an extended family.} 
	  \label{fig:kenyaNetwork}
	  \vspace{12mm}
\end{minipage}
\quad
\begin{minipage}[b]{0.65\linewidth}
	\begin{subfigure}{.48\textwidth}
		\includegraphics[width=1\linewidth]{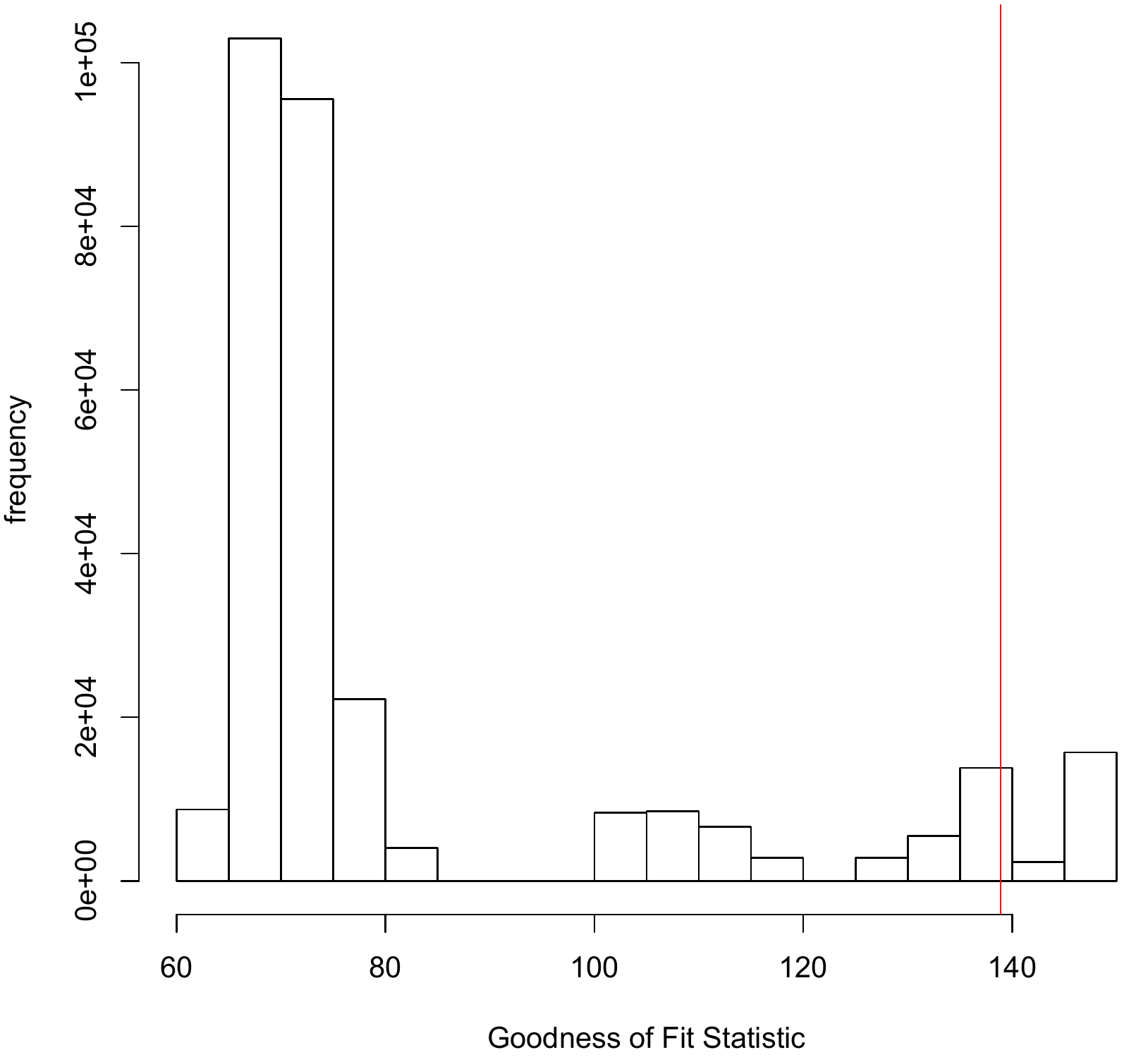}
	  \caption{Histogram for  sampling distribution of  chi-square statistic, with indicated observed value.} 
	  \label{fig:kenyaGofs}
	\end{subfigure}
	\quad
	\begin{subfigure}{.5\textwidth}
		\includegraphics[width=1\linewidth]{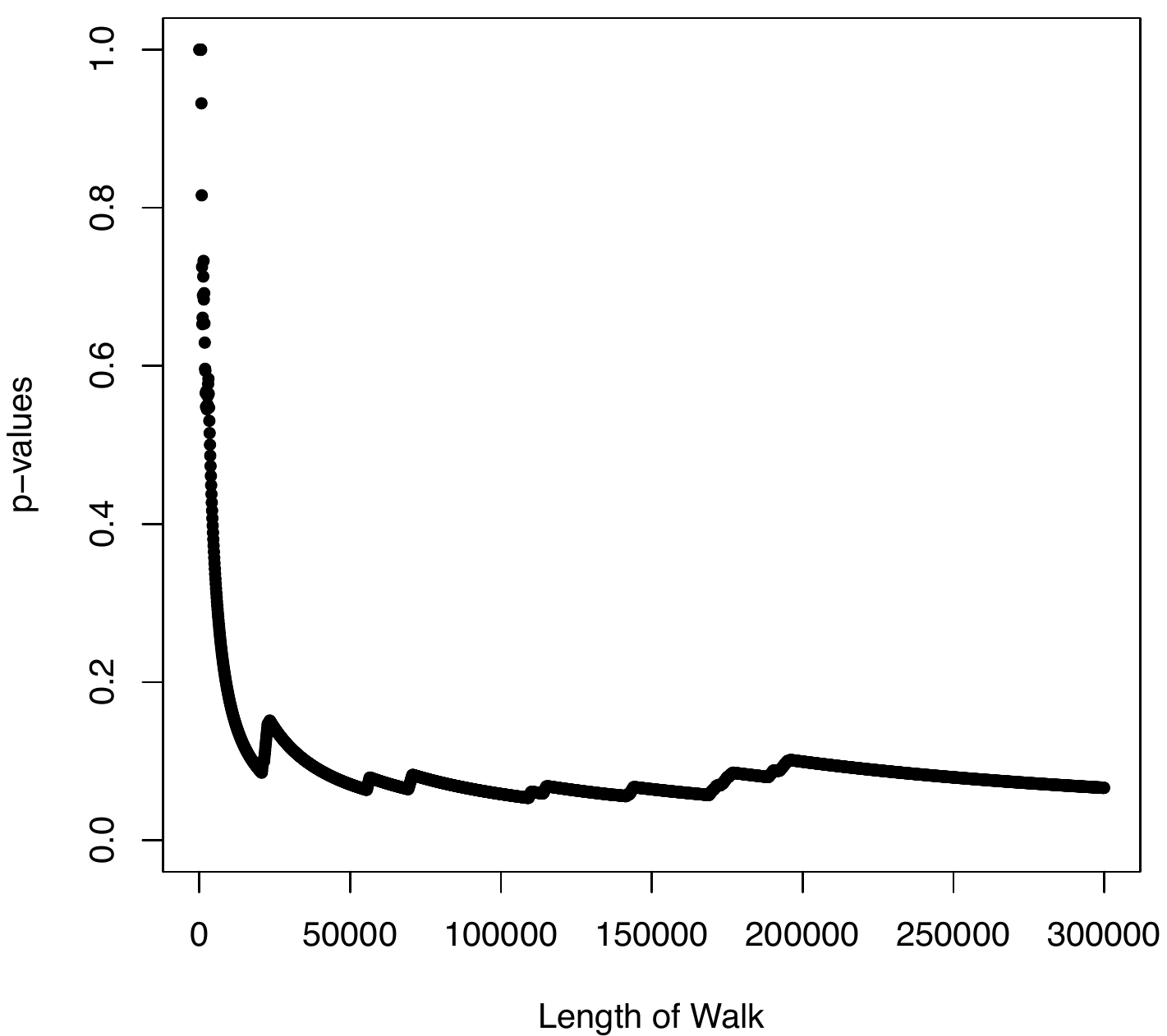}
	  \caption{The $p$-value estimates.}
	  \label{fig:kenyaPvalues}
	\end{subfigure}
	\caption{Simulation results for the mobile money network data  from \cite{KCGK13}: chain of length $330,000$ with $30,000$ burn-in steps.}
\end{minipage}
\end{figure}

Running Algorithm~\ref{alg:MH} for 300,000 steps after an initial burn-in of 30,000 steps returns an estimated p-value of $0.06024261$, which would suggest that the $p_1$ model with edge dependent reciprocation is indeed a poor fit for this data, and in fact,  if the significance level is set to less than $0.1$ we would reject the model.
Figure \ref{fig:kenyaGofs} shows the histogram of the sampling distribution of the chi-square  statistics with the chi-square  statistic for the observed network marked in red. Figure \ref{fig:kenyaPvalues} shows the estimated p-value plotted against the step number of the Markov chain and gives evidence of convergence.

\subsection{Chesapeake Bay Ecosystem}
\label{sec:ex:bay}

In their 1989 paper \cite{BayData}, Ulanowicz and Baird constructed trophic networks for specific regions of the Chesapeake Bay using extensive data gathered from 1983-1986. Their work used highly sophisticated estimation methods, relying on a multitude of different sources.  Due to their profound detail, Ulanowicz and Baird's food webs have been extensively analyzed over the last 25 years.  Often for statistical model-fitting purposes, the edges are considered as undirected.  This choice, however, has been largely motivated by the scarcity of tools available to analyze directed networks.  Other than heuristic methods, procedures for performing goodness of fit testing for directed network models have not existed.

The data set on which we test the $p_1$ model  is depicted in Figure~\ref{fig:BayGraph}; see also \cite[Figure 2]{BayData}. 
  The list of edges of this directed network was downloaded from \cite{BayPajek} and  represents the Web 34 Chesapeake Bay Mesohaline Ecosystem.  The graph has 39 vertices and 176 edges.  The majority of vertices represent species in a Chesapeake Bay food web, 
with a directed edge $u\to v$ indicating that species $u$ eats species $v$. Although, we note there are also other elements, which are not species, included as vertices as well, such as passive carbon storage compartments. There are $6$ reciprocated edges in the graph. 

We expect a block structure in food networks that do not naturally occur in $p_1$-model generated networks. In fact, the estimated p-value is $0.03459158$, indicating that  the $p_1$ model with edge dependent reciprocation is not a good for this data. If the significance level is set to less than $0.05$, we would reject this model. The histogram of a simulation with $1,000,000$ steps is shown in Figure~\ref{fig:Bay-1MgofsHistogram}.

\begin{figure}
  \includegraphics[width=.5\linewidth]{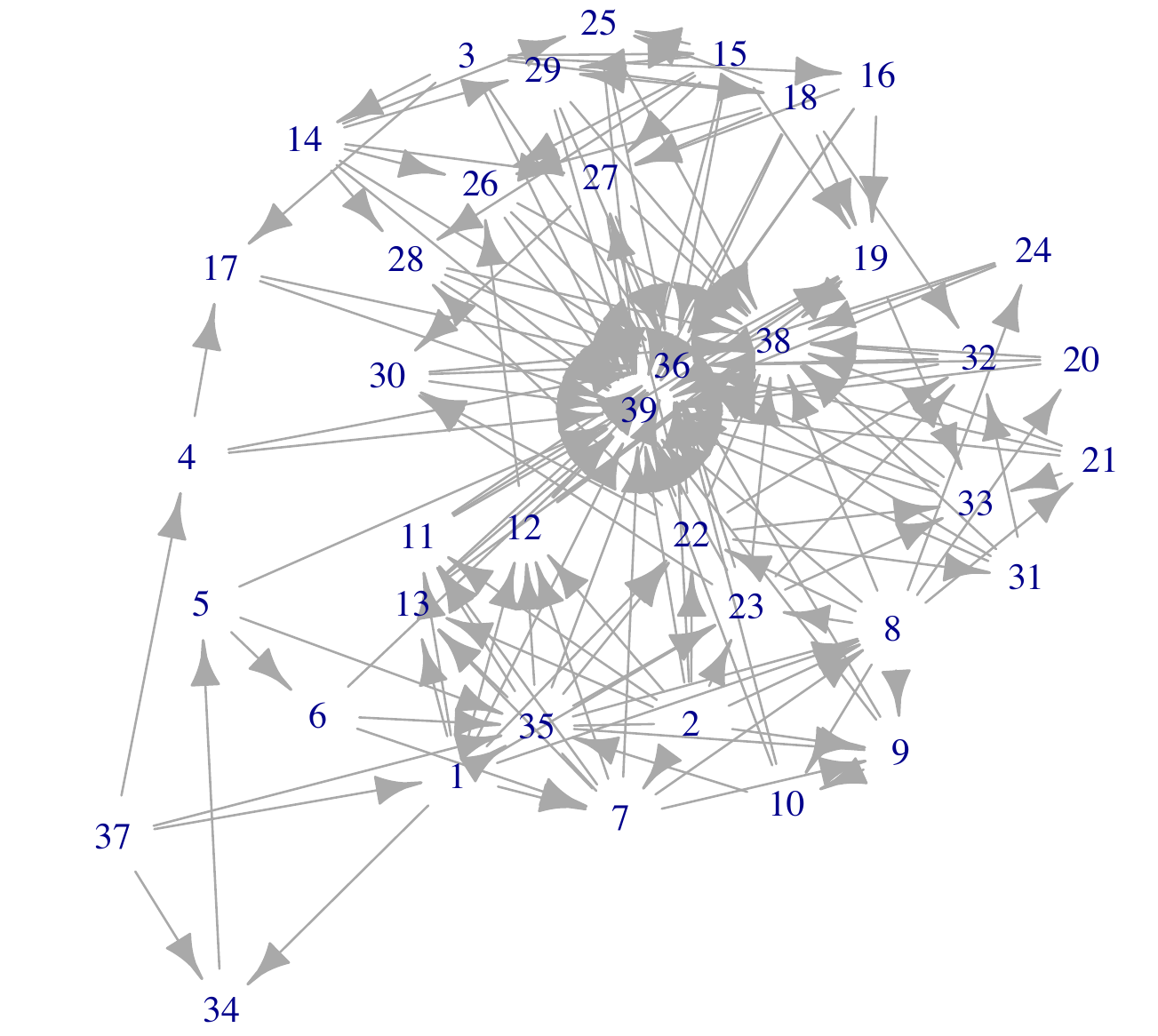} 
  \caption{The directed network representing the food web relationships in Chesapeake Bay data from \cite{BayData}. 
  		} 
\label{fig:BayGraph}
\end{figure}

\subsection{Sampson's Monastery Study}
\label{sec:monk}

Sampson \cite{Sampson68} conducted an ethnographical study of social interactions between novices in a New England monastery in the mid 1960s.   Sampson observed $25$ novices over a period of two years, gathering social relations data at 4 time points, and on multiple relationships. 
This has been a favorite example for analysis by sociologists, statisticians and others, and was used in original $p_1$ model studies. At the fourth time point ($T4$), there were $18$ monks, and the social network had 54 directed edges representing the top three answers  to the question ``whom do you like'' for each novice.  
We consider the directed graph in  Figure~\ref{fig:MonkGraph} representing the  relationships derived from this affinity sociometric data. The list of edges in the graph was downloaded from \cite{MonkPajek}.

\begin{figure}
\begin{subfigure}[b]{.5\textwidth}
	 \centering
	\includegraphics[width=\linewidth]{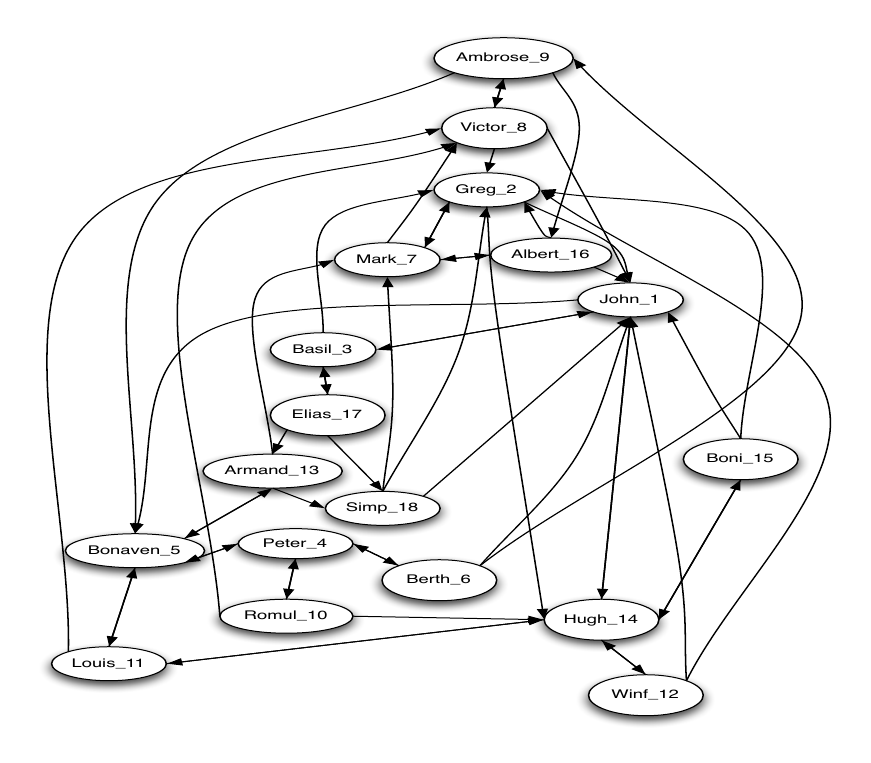} 
	\caption{\cite[Figure~2.1]{F-review}}
\end{subfigure}
\caption{Network derived from the monk dataset at time $T4$ in \cite{Sampson68}.} 
\label{fig:MonkGraph}
\end{figure}

Perhaps  not unsurprisingly, the $p_1$-model with edge-dependent reciprocation seems to fit this data remarkably well. The chi-square statistic for the observed network is 404.7151, which is very close to the minimum chi-square statistic that was returned during a 1,000,000 step walk (see Figure~\ref{fig:monkGOF}).  The estimated $p$-value for this data is $0.9863126$.
The random walk seems to be exploring the fiber broadly, discovering about $8800$ new networks every $50,000$ steps, though we do not know the exact size of the fiber. 
\begin{figure}
\begin{subfigure}{.55\textwidth}
  \centering
  \includegraphics[width=.9\linewidth]{finalplot-Monk-1MgofsHistogramExcl50kBurninSteps.pdf}
  \caption{Histogram of chi-square values from the simulation.}
  \label{fig:monkGOF}
\end{subfigure}
\quad
\begin{subfigure}{.4\textwidth}
  \centering
  \includegraphics[width=1\linewidth]{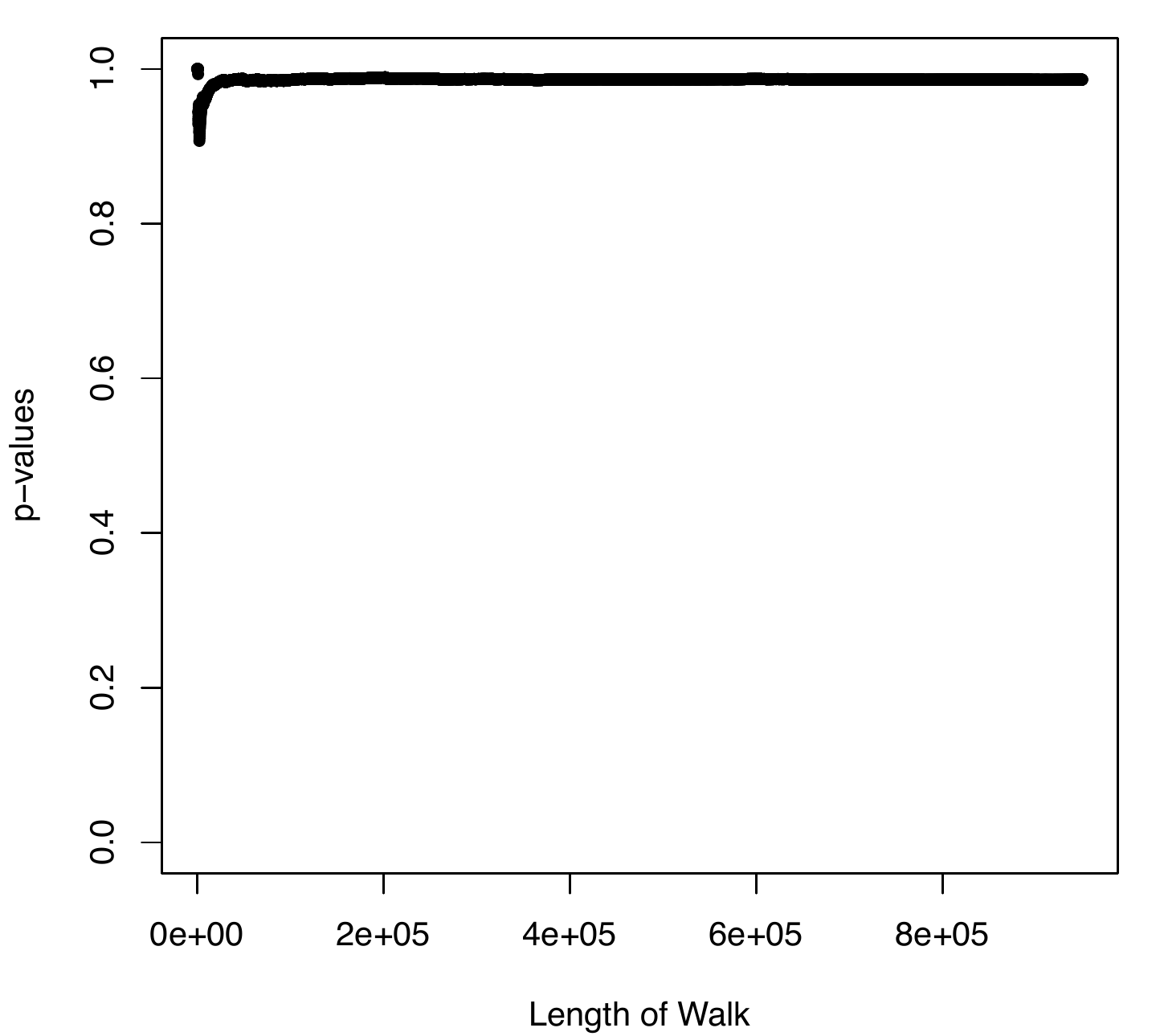}
  \caption{The $p$-value estimates.}
  \vspace{-12mm}
\end{subfigure}
\caption{Results of simulation with 1 million steps (50,000 burn-in steps) for  Sampson's monastery data from time period $T4$.} 
\label{fig:monkAnalysis}
\end{figure}

\section{Conclusion}

The central motivation for this work is the scarcity of tools available to analyze directed networks.  Other than heuristic methods, procedures for performing goodness of fit testing for directed network models have not existed.
 In the usual setting, the  Metropolis-Hastings algorithm for sampling from conditional distributions requires a Markov basis for a given model to be precomputed. By definition, however, Markov bases are data independent, thus presenting a computational problem that becomes both wasteful and infeasible for  network models on as few as $7$ nodes. 
In addition, sampling constraints (e.g. one edge per dyad in a network or cell bounds in a contingency table) have presented problems for algebraic statistics as the restricted (observable) fibers cannot always be connected with a minimal set of Markov moves. Instead, a knowledge of a much larger set of moves, such as the Graver basis, is required for sampling. 
Since Graver bases are notoriously difficult to compute except for (notable)  special cases (e.g. where a divide-and-conquer strategy applies, as in decomposable models), being able to dynamically generate one applicable move at a time is essentially the only hope for ever being able to utilize the algebraic statistics idea in practice. 

 Using the work by Dobra \cite{Dobra2012} as our main motivation,  we propose a methodology for dynamically generating  moves and combinations of moves from the Graver basis (and thus a Markov basis) that guarantee to connect observable fibers for networks or contingency tables where sufficient statistics are not necessarily table marginals.  
This approach allows for a data-oriented algorithm, providing a dynamic exploration of any fiber without relying on an entire Markov basis.  It produces only a relatively small subset of the moves - which could still be a large subset indeed - needed to connect the observable points in the fiber.  

In contrast with previous approaches, our proposed modification uses moves that are constructed by understanding the balanced edge sets of the parameter hypergraph of the given model.  Drawing upon the classical literature in combinatorial commutative algebra and recent work in algebraic statistics, we show how, in principle, one can construct applicable moves using the parameter  hypergraph of any log-linear model and any observed network. 
Thus, in situations where the structure of the parameter hypergraph is well understood, this allows for easily implementable algorithms for goodness-of-fit testing.  
As an example, we describe the entire procedure on the $p_1$ model with edge-dependent reciprocation.  For the $p_1$ model, we (1) derive the structure of such the Markov moves in relation to the parameter hypergraph and (2) implement an algorithm to generate them dynamically.   
We hope this technique of analyzing the parameter hypergraph to construct dynamic Markov bases will be used for other log-linear models and spurs new ideas for goodness-of-fit testing for exponential random graph models in general.

\section*{Acknowledgements}
The authors are grateful to Alessandro Rinaldo and Stephen E. Fienberg for their support at the inception of this project. The first author is supported by the NSF Postdoctoral Research Fellowship, NSF award \#DMS-1304167. 
The second and third authors acknowledge partial support from  grant \#FA9550-12-1-0392 from the U.S. Air Force Office of Scientific Research (AFOSR) and the Defense Advanced Research Projects Agency (DARPA). 
Some computations are performed on a cluster provided by an NSF-SCREMS grant to IIT. 

\bibliography{DynMarkP1}
\bibliographystyle{amsalpha}   
\end{document}